%% file: main.tex
\documentclass[12pt]{article}

\usepackage[T1]{fontenc}
\usepackage{tgpagella}

\usepackage[english]{babel}

\usepackage[margin=1in]{geometry}

\usepackage{amsmath,amsthm,amssymb,amsfonts}
\usepackage{mathtools}
\usepackage{graphicx}
\usepackage{hyperref}
\usepackage{algorithm, algpseudocode}
\usepackage{multirow}
\usepackage[backend=biber,style=numeric,maxbibnames=999,maxcitenames=3]{biblatex}
\usepackage{csquotes}
\usepackage[T1]{fontenc}
\usepackage{parskip}
\usepackage{pdflscape}
\usepackage{xcolor}
\usepackage{siunitx}
\usepackage{makecell}
\usepackage[normalem]{ulem}
\usepackage{adjustbox}
\usepackage{caption}
\usepackage{booktabs}
\usepackage{afterpage}
\usepackage{float}
\usepackage{enumitem}
\usepackage{url}
\usepackage[most]{tcolorbox}

\input{macros}

\title{Differentially Private Release of\\Israel's National Registry of Live Births}

\author{Shlomi Hod\footnotemark[1] \and Ran Canetti\thanks{Department of Computer Science, Boston University. \,\texttt {\{shlomi,canetti\}@bu.edu}\\Supported by DARPA under Agreement No. HR00112020021.}}

\begin{document}

\maketitle

\input{sections/0-abstract}

\newpage

\tableofcontents

\newpage

\input{sections/1-introduction}
\input{sections/2-components}
\input{sections/3-related-work}
\input{sections/4-discussion}
\input{sections/5-acknowledgement}

\printbibliography

\end{document}

%% file: macros.tex
\renewbibmacro*{textcite}{
  \printnames{labelname}
  \addspace
  \printfield[parens]{year}
  \setunit{\addnbspace}
  \printtext[brackets]{\bibhyperref{\thefield{labelnumber}}}
}

\hypersetup{colorlinks=true, allcolors=blue}

\tcbset{
    box/.style={
        colback=white,
        colframe=black,    
        coltitle=white,
        colbacktitle=black,
        fonttitle=\bfseries,
        top=8pt,
        bottom=8pt,
        left=10pt,
        right=10pt,
        arc=8pt,
        boxrule=1pt,
        after skip=12pt
    }
}

\DeclareMathAlphabet{\mathpzc}{OT1}{pzc}{m}{it}

\DeclarePairedDelimiter\floor{\lfloor}{\rfloor}

\newcommand{\eps}{\varepsilon}

\newcommand{\E}{\mathbb{E}}
\newcommand{\R}{\mathbb{R}}

\newcommand{\DD}{\mathcal{D}}
\newcommand{\DR}{\mathcal{R}}
\newcommand{\DS}{\mathcal{S}}
\newcommand{\XX}{\mathcal{X}}
\newcommand{\EE}{\mathcal{E}}

\newcommand{\M}{\mathpzc{M}}

\newcommand{\AC}{\mathsf{AC}}
\newcommand{\REC}{\mathrm{Rec}}

\newcommand{\mathbbm}[1]{\text{\usefont{U}{bbm}{m}{n}#1}}

\newcommand\inv[1]{#1\raisebox{1.15ex}{$\scriptscriptstyle-\!1$}}

\makeatletter
\newcommand\footnoteref[1]{\protected@xdef\@thefnmark{\ref{#1}}\@footnotemark}
\makeatother

\newtheorem{definition}{Definition}
\newtheorem{theorem}{Theorem}
\newtheorem{proposition}{Proposition}
\newtheorem{lemma}{Lemma}
\newtheorem{remark}{Remark}

\newcommand{\defeq}{\vcentcolon=}

%% file: sections/0-abstract.tex
\begin{abstract}
In February 2024, Israel's Ministry of Health released microdata of live births in Israel in 2014.
The dataset is based on Israel's National Registry of Live Births and offers substantial value in multiple areas, such as scientific research and policy-making, while providing pure differential privacy guarantee with $\eps = 9.98$ for 2014's mothers and newborns. 
The release was co-designed by the authors along with stakeholders from both inside and outside the Ministry of Health.
This paper presents the methodology used to obtain that release, which, to the best of our knowledge, is the first of its kind in the world.

The design process has been  challenging and required flexibility and open-mindedness on all sides involved, along with substantial technical innovation. In particular, we introduce new concepts regarding the desiderata from dataset releases in a microdata format, as well as a way to bundle together multiple quantitative desiderata for a differentially private release using the private selection algorithm of Liu and Talwar (STOC 2019). We hope that the experiences reported here will be useful to future differentially private releases.
\end{abstract}

%% file: sections/1-introduction.tex
\section{Introduction}

The Israeli National Registry of Live Births is the official source of all live births in Israel. According to officials in the Ministry of Health, it holds data that is extremely valuable for research and policy-making in many areas, including demography, economy and health. Making the Registry accessible to the public has a significant value in public health and policy development.

In February 2024 the Israeli Ministry of Health released selected data fields (columns) from the Registry of singleton births (births with only a single newborn) from 2014 \cite{moh2024release}. The data was processed so as to protect the privacy of mothers and newborns. The Ministry of Health is treating this release as a pilot, and anticipates releasing additional fields and data from more recent years based on feedback from the project's stakeholders. More broadly, the project aims to demonstrate the effectiveness of Differential Privacy (DP) \cite{Dwork2006CalibratingNT} for publishing government data \cite{Drechsler2021ForGovernment}, as opposed to traditional, deterministic anonimization methods (as initially planned prior to our engagement with the ministry).

The Registry comprises 167K records of all singleton live births in Israel during 2014 (The total number of births is public information, provided by the Israeli Central Bureau of Statistics \cite{cbs2014newborns}.) Six data fields per live birth were chosen to be included in the release because of their potential values to various data users (see Table~\ref{tab:dataset}).
The release was co-designed by the authors together with several stakeholders, including TIMNA (Israel’s National Health Research Platform), which managed the project. This paper documents the process leading to this release, presenting the requirements elicited from the stakeholders, some of which, to the best of our knowledge, are described here for the first time in the literature. We provide technical definitions of these requirements along with methods to satisfy them. Additionally, we discuss the considerations behind our design choices and our scheme to overcome core challenges in producing a privacy-protected dataset release.

The Registry contains national-level sensitive information about mothers and newborns. The release of such sensitive data is regulated by Article 7 of the of the Basic
Law: Human Dignity and Liberty (1992) \cite{IsraelHumanDignityandLiberty1992}, 
Israel's Patient's Rights Law (1996) \cite{IsraelPatientRightsLaw1996}, and Protection of Privacy Law (1981) \cite{IsraelPrivacyLaw1981}. Indeed, releasing national birth data without taking appropriate privacy-protection measures could lead to serious privacy harms as discussed, e.g., in  \cite{solove2008understanding,Citron2021PrivacyH,hartzog2021privacy}.

Before going public, the birth dataset release, alongside its documentation, went through a scrutinizing process of reviews and approvals by multiple entities within the Ministry including TIMNA (Israel’s National Health Research Platform) the National Birth Registry steward, the Chief Data Officer and the Ministry of Health’s legal department. This process also includes a broader analysis of the release, weighing the privacy risks under the protection measures taken, against  the benefits to society, such as improving public health by enabling scientific research and supporting evidence-based policy-making.

\begin{table}[H]
\centering
\caption{{\bfseries The metadata of the released dataset of singleton live births in 2014\\$(n = \num[group-separator={,}]{165915})$.}}
\begin{adjustbox}{width=\textwidth,center}
\begin{tabular}{@{}lll@{}} \toprule
Column    & Description & Possible Values 
\\
\midrule
\texttt{birth\_month} & Month of birth in 2014   & \makecell{1, 2, 3, 4, 5, 6, 7, 8, 9,\\10, 11, 12}
\\ \addlinespace
\texttt{mother\_age} & Age of mother in full years at birth  & \makecell{<18, 18-19, 20-24, 25-29, 30-34,\\35-36, 37-39, 40-42, 43-44, 44<}
\\ \addlinespace
\texttt{parity} & Mother's number of live births so far & \makecell{1, 2-3, 4-6,\\7-10, 10<}
\\ \addlinespace
\texttt{gestation\_week} & The week of pregnancy when the birth took place & \makecell{<29, 29-31, 32-33, 34-36,\\37-41, 41<}
\\ \addlinespace
\texttt{birth\_sex} & Sex assigned at birth  & \makecell{M (male),\\F (female)}
\\ \addlinespace
\texttt{birth\_weight} & Newborn weight at birth in grams  & \makecell{<1500, 1500-1599, \ldots,\\4400-4499, 4499<} \\
\bottomrule
\end{tabular}
\end{adjustbox}
\label{tab:dataset}
\end{table}

\subsection{Requirements and Solution Concepts}

In this section, we present the requirements, some of which evolved during the project, set by the stakeholders listed in the following box, along with our solution concept for fulfilling them.

\begin{tcolorbox}[box, title=Stakeholders]
    \begin{enumerate}[leftmargin=*,itemsep=4pt]
        \item TIMNA, Israel's National Health Research Platform, Ministry of Health: \textbf{the primary stakeholders}
            \begin{enumerate}[leftmargin=*,itemsep=2pt]
                \item Epidemiology team lead
                \item De-identification team lead 
            \end{enumerate}
        \item Other organizations in the Ministry of Health
            \begin{enumerate}[leftmargin=*,itemsep=2pt]
                \item Chief Data Officer
                \item Head of data stewardship department
                \item Birth Registry steward
            \end{enumerate}
        \item Universities and Research Institutes
            \begin{enumerate}[leftmargin=*,itemsep=2pt]
                \item Biostatistician researcher with expertise in birth data
                \item Pediatrician and medical researcher
            \end{enumerate}
    \end{enumerate}
\end{tcolorbox}

\subsubsection{Privacy}
\label{sec:intro-privcay}

The privacy of individuals within this release is protected by differential privacy guarantees \cite{Dwork2006CalibratingNT}.
Differential privacy is a mathematical framework that limits the contribution of individual records to an analysis or computation, thus restricting what can be learned about those specific records.
In a nutshell, an algorithm that outputs the result of some statistical computation, applied to a given dataset, provides differential privacy if the influence of any single record from the dataset on the algorithm's output is controlled. Differential privacy quantifies this influence with a parameter often called \emph{privacy loss budget} and denoted $\eps$. 
While 
differential privacy has been in rapidly growing use, including in data releases by the US government \cite{Abowd2018TheUC,miklau2022negotiating,burman2019safely}, big tech companies \cite{Wilson2019DifferentiallyPS,fitzpatrick2020} and other organizations \cite{Adeleye2023Wikimedia,Migration2022Global}, we are not aware  of a use of differential privacy for any institutional release of public health data.
The privacy loss budget used in this release is set to $\eps = 9.98$. Although this is high value from a theoretical perspective, it falls within the mid-range of past releases \cite{desfontain2021ListRealworld}, and evidence from the literature suggests that it provides protection against practical adversaries \cite{Stadler2020SyntheticD}. In Section~\ref{sec:components-differential-privacy}, we discuss how the budget was set and explore other design choices related to differential privacy.

\subsubsection{Format}
\label{sec:intro-format}

The most common approach to apply differential privacy in practice is to enumerate a concrete list of queries of interest to the data users (e.g., specific histograms and contingency tables), compute them with differential privacy and publish the results. For this strategy we have the most developed theory and tools to carefully balance between privacy protection and accuracy.

When we proposed this solution concept to the primary stakeholders, they have rejected it strongly. The stakeholders envisioned a release in the format of microdata (tabular data), where each record represents an individual birth.

In light of this requirement and the choice of differential privacy as privacy protection measure, we proposed to produce the release with \emph{synthetic data}. 
We stress that our proposal of synthetic data is not as privacy protection measure as suggested by others \cite{rubin1993statistical,Nowok2016synthpopBC,Drechsler2023ThirtyYears,Garfinkel2023},
but rather as a design choice resulting from the stakeholders' requirement to have a microdata format \cite{Stadler2020SyntheticD}. 
In past real-world deployments, differentially private synthetic data is a rare solution concept \cite{desfontain2021ListRealworld}. Synthetic data is typically generated by sampling records from a statistical generative model trained on the original dataset \cite{Jordon2022SyntheticD}. 

\subsubsection{Quality}
\label{sec:intro-quality}

One core challenge in using differentially private synthetic data is the lack of tools in the literature to establish accuracy and quality guarantees for a diverse collections of queries: some of the synthetic data mechanisms includes theoretical quality guarantees for a collection of statistical queries (such as MWEM for contingency tables \cite{Hardt2010ASA}) and some are not (such as PATE-GAN \cite{Jordon2019PATEGAN} and PrivBayes \cite{Zhang2014PrivBayesPD}). However, these theoretical quality guarantees might be insufficient because they are either too loose or not aligned with the intended use of the synthetic data \cite{arnold2020}.

Together with the stakeholders we elicited diverse families of queries that the released dataset is primarily designed to answer: contingency tables
($k$-way marginals), means with group-by aggregations (conditional means), median with group-by aggregations (conditional medians) and linear regressions. According to our stakeholders, contingency tables and central tendency queries (means, medians) carry a lot of value.

To guarantee the quality of the released dataset with respect to these queries, we set a collection of \emph{acceptance criteria} that are assessed empirically on the released dataset.
Passing an acceptance criterion means that the a statistical query executed on the released dataset is sufficiently close, as defined by the stakeholders, to the same query applied on the original dataset.
The acceptance criteria, as used in our overall solution, are also a response to the gap in the literature mentioned above and provide a tunable way to ensure accuracy and quality guarantees.

\subsubsection{Configurations}
\label{sec:intro-configurations}

One of the first steps in creating a synthetic dataset is selecting the \emph{family} of generative mechanisms to utilize. Typically, mechanisms ensuring differential privacy involve running a learning algorithm to produce a generative model that can subsequently be sampled in order to generate synthetic data. Options include marginals-based models, GANs, and Bayesian networks. 
In addition, these learning algorithms and the resulting generative models tend to have \emph{hyperparameters}, namely parameters set by the user to control the learning process and the model structure. For example, consider the learning rate for GANs or maximal node degree in a Bayesian network. Depending on the data, some hyperparameters might work well and others might lead to models that produce low-quality synthetic data.

Additionally, we allowed the stakeholders to specify multiple alternatives for representing data fields. In other words, the original dataset can undergo various transformations to adjust the granularity of each data field through binning, as predefined by the stakeholders. 

We define a \emph{configuration} as the complete specification required for generating the released dataset, including the model family, hyperparameters and data transformations. We span a Cartesian product of all possible combinations. In Section~\ref{sec:intro-scheme}, we present our approach to find sufficiently-good configuration to produce the released dataset.

\subsubsection{Faithfulness}
\label{sec:intro-faithfulness}

Initially, the primary stakeholders rejected the synthetic data solution, even though it would accompanied with quality guarantees via the acceptance criteria, because synthetic data has a negative reputation within the medical research community. The stakeholders were not familiar with successful usage of synthetic data for statistical analysis. In addition, according to the primary stakeholders, synthetic data might raise a transparency concern when it is officially released by the government. Some might consider such data as ``fake'', and treat it as an attempt of the government to ``hide'' information. At that point of the project, it seemed that there is no way forward to satisfy these requirement of differential private release in a microdata format \emph{without} using synthetic data.

Through a deliberated dialogue with the stakeholders, we were able to pin down the reason behind these statements: \emph{the lack of record-level alignment between the original and synthetic datasets}. To address this concern, we defined a new technical notion called (record-level) \emph{faithfulness}. Faithfulness assesses whether the released dataset resemble the original dataset in a record-level granularity.

The stakeholders have set the minimal accepted level of faithfulness, and this requirement was included as an additional acceptance criterion.

The introduction of the faithfulness criterion effectively alleviated discomfort about using synthetic data. We observed this on multiple occasions with stakeholders from diverse backgrounds and expertise. We hypothesize that while faithfulness may not provide an additional notion of accuracy beyond the quality acceptance criteria, it enhances the perceived trustworthiness of the release.

\subsubsection{Face Privacy}
\label{sec:intro-face-privacy}

In a later stage of the project the head of data stewardship department and the Birth Registry steward introduced an additional requirement à la $k$-anonymity \cite{Sweene02}.
For each column $c$ and each set $B$ of $b$ other columns, treat $c$ as a sensitive column, treat the columns in $B$ as
quasi-identifier columns and require that this configuration satisfies $k$-anonymity. In other words, for all possible choices of values for the columns in $B$, each value configuration appears in at least $k$ rows. The requirement was made with $k=5$ and $b=5$, namely one less than the overall number of columns. This is a strong requirement which, to the best of our knowledge, was not presented before in the literature even with respect to non-synthetic data.

A synthetic dataset sampled from a differential private generative model enjoys formal privacy guarantees. Nonetheless, stakeholders, especially data subjects, might have a different expectation of how privacy-protected data might look like. We call such expectations \emph{face privacy}, regardless of whether these expectations contribute to mitigating privacy risks such as reconstruction and attribute inference attacks. The term is inspired by the term ``face validity'' from Psychometrics. Face validity is the appropriateness or sensibility of a test as they appear to the test-taker. It is subjective and distinct from the technical notion of validity---whether a test measures what it is supposed to measure \cite{Irving2010TheCE}.

Following a discussion with the stakeholders, the requirement converged into this face privacy statement: \emph{private governmental data released to the public cannot contain unique records under any circumstances, even if it is synthetic data generated with differential privacy}. 

To achieve the ``no unique rows'' expectation in this release, we apply a specific \emph{dataset projection} on the synthetic dataset.
The output of this projection also satisfies complete $k$-anonymity.
The projection only alters the synthetic dataset without sampling additional rows from the generative model. From a differential privacy prospective, the dataset projection is a post-processing step, so it does not consume any privacy loss budget.
The output dataset of the projection is evaluated against the acceptance criteria.

\subsubsection{Transparency}
\label{sec:intro-transparency}

Since the released data is presented in a tabular format, it has the same affordance\footnote{Following \textcite{Norman2013TheDO}, we use the term \emph{affordance} to represent the set of possibilities for using the data, as perceived by the data user .} as the original data. This may mislead the data users and give them the impression that they can perform any type of analysis on it.
However, it is important to note that the released dataset has been specifically designed to address predefined, albeit broad, statistical queries outlined in the acceptance criteria. We cannot guarantee the quality or accuracy of the results obtained from other types of analysis, such as more complex machine learning tasks or anomaly detection.
This is issue is inherent to differentially private microdata, as shown by \textcite{UllmanV11}, it is in general impossible to efficiently generate microdata that is differentially private while also preserving the accuracy of all statistical queries \cite{UllmanV11}. In fact, even preserving the accuracy of all two-way marginals is impossible. 

To mitigate this risk of inappropriate analysis, we rely on the README document as our primary method of communication, particularly through the inclusion of the \emph{intended usage} section. The documented is treated as inseparable part of the release.

\begin{tcolorbox}[box, title=Summary of stakeholders' requirements]
    \begin{enumerate}[leftmargin=*,itemsep=4pt]
        \item \textsc{Format}: The data should be presented in tabular form of individual records.
        
        \item \textsc{Quality}: Accuracy with respect to a broad set of statistical measures.
        
        \item \textsc{Faithfulness}: There should be row-level mapping between the released and the original data.
        
        \item \textsc{Privacy}: Protecting the privacy of mothers and newborns. This requirement came in two flavors: first, differential privacy as rigorous state-of-the-art privacy criteria must be met; in addition, some face privacy requirements were made (e.g., unique rows are not allowed).
        
        \item \textsc{Transparency}: The intended and unintended uses of the dataset should be clearly documented, particularly detailing which statistical queries carry accuracy guarantees and which are not.
    \end{enumerate}
\end{tcolorbox}

\subsection{Universal Scheme for Differentially Private Synthetic Data}
\label{sec:intro-scheme}

So far we have discussed the release's requirements one by one, along with their corresponding solution concepts. Now we turn to present our universal scheme that integrate all the solution concepts together and addresses additional technical challenges in deploying differentially private synthetic data (and, in generally, any differentially private release.) The universal scheme is illustrated in Figure~\ref{fig:scheme} and summarized in 
Algorithm~\ref{alg:scheme}. The description of its concrete steps follows.

We assume that the dataset provided as input to the scheme, referred to as the \emph{original dataset}, has been preprocessed and cleansed from the \emph{raw dataset}, the dataset extracted from the registry.

In the scheme, we first apply the data transformations, as defined in the configuration, on the original dataset to obtain the \emph{transformed dataset}, which is used to train a generative model. We then sample records from this model, filtering them based on predefined plausibility constraints to create a \emph{synthetic dataset} (see Section~\ref{sec:components-constraints}). The training process includes various configurations, such as model family selection, and hyperparameter tuning (see Section~\ref{sec:components-configurations}).

We then perform a dataset projection on the synthetic dataset, creating a \emph{release-candidate dataset}. This is done to guarantee face privacy (see Section~\ref{sec:components-projection}).

Finally, we compare this release-candidate dataset to the transformed dataset, using a set of \emph{acceptance criteria} to gauge quality (see Section~\ref{sec:components-acceptance-criteria}). 
One criterion evaluates record-level faithfulness, ensuring a one-to-one row similarity between the two datasets (see Section~\ref{sec:components-faithfulness}). Each additional criterion comprises an error measure and a corresponding threshold. The release-candidate dataset must meet these criteria---based on error measures for stakeholder-specific statistical queries---to be accepted.

To preserve the differential privacy guarantee, all computations involving the original data should adhere to differential privacy to ensure the desired level of privacy protection. This not only includes the step of learning the generative model but also extends to additional steps that may rely on private data---such as data transformations, model selection, hyperparameter search, post-processing, and evaluation. If these steps are not conducted under differential privacy, the resulting release may not be differentially private or may offer ineffective privacy protection \cite{Liu2018PrivateSF,Papernot2022Hyperparameter,Xiang2024How}.

\begin{figure*}[tb]
    \centering
    \includegraphics[width=\textwidth]{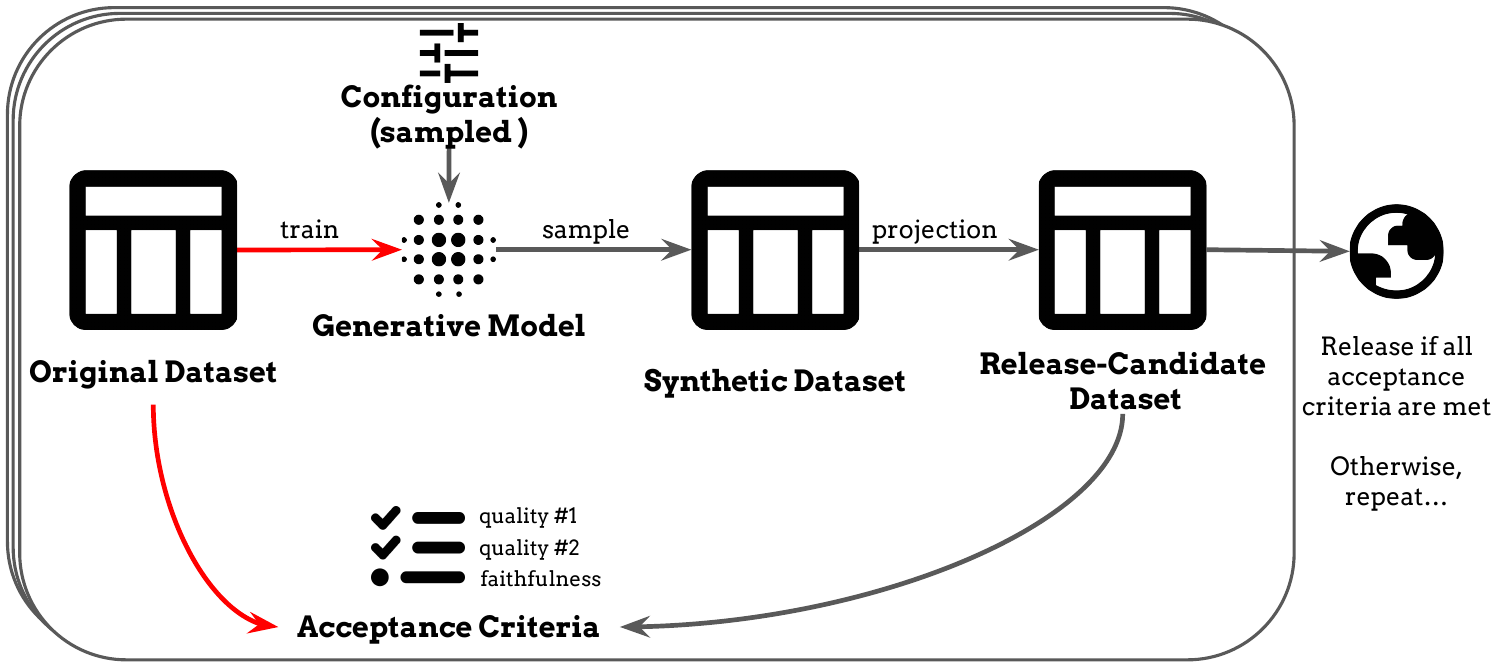}
    \captionof{figure}{Schematic overview of the universal scheme for  differentially private microdata data release. A red arrow ({\color{red}$\rightarrow$}) represents a computation done with differential privacy.}
    \label{fig:scheme}
\end{figure*}

Once the scheme generates a dataset that meets all acceptance criteria, we publicly release that dataset, \emph{along with a differentially private release of the actual accuracy measures for that particular dataset.}  

Conducting multiple iterations to search over configurations would result in a linear accumulation of the privacy loss budget if done straightforwardly with simple composition. The private selection algorithm with a known threshold (Algorithm~\ref{alg:known-threshold}) of \textcite{Liu2018PrivateSF} addresses this issue by allowing multiple iterations until success, while consuming only twice the privacy loss budget of a single iteration, unlike the factor dependent on $T$ attempts with simple composition \cite{Liu2018PrivateSF} (see Section~\ref{sec:components-private-selection}).

The fact that we output an estimate of the actual  accuracy measure of the released dataset is crucial: It allows us to set the hard threshold of the selection algorithm (Algorithm~\ref{alg:known-threshold}) at a relatively low level, thus reducing the probability of failure to generate a dataset within reasonable runtime. Since these estimates are released with differential privacy  guarantees, the overall privacy of the original data is preserved. 

The scheme is \emph{universal} in the sense that it designed to be algorithm-agnostic and readily-used, regardless of the generative model and the collection of acceptance criteria. The scheme takes as input any ``standalone'' differentially private synthetic data generators. These generators often lack quality guarantees or provide some guarantees over a specific type of queries (e.g., marginals) only. In contrast, our universal scheme uplift any differentially private synthetic data generator to produce synthetic data that fulfil any required empirical guarantees as long as the guarantees (1) are  quantitatively expressed with differential privacy; and (2) are feasible given the input dataset and configuration space.

\begin{algorithm}[tb]
    \caption{Universal scheme for differentially private microdata data release.}\label{alg:scheme}
     \begin{algorithmic}[1]
        \Require Original dataset $\mathcal{O}$, privacy loss budget for generative model $\eps_x > 0$, privacy loss budget for acceptance criteria $\eps_q > 0$, list of experiment configurations $\mathcal{C}$, list of constraints $\mathcal{F}$, list of acceptance criteria $\mathcal{A}$
        \While{True}
            \State Uniformly sample a configuration $c$ from $\mathcal{C}$
            \State Create transformed dataset $\DR$ using data transformations specified by $c$
            \State Fit an $\eps_x$-DP generative model $G$ using hyperparameters from $c$
            \Repeat
                \State Sample a synthetic record $s$ from $G$
                \State Retain $s$ if it satisfies all constraints in $\mathcal{F}$; otherwise, discard it
            \Until{Collected $|\DR|$ synthetic records}
            \State Apply dataset projection specified by $c$ to obtain release-candidate dataset $\DS$
            \State Compute $\eps_q$-DP acceptance criteria $\mathcal{A}$ between $\DR$ and $\DS$
            \State Store the acceptance criteria results in $Q$
            \State Test if all acceptance criteria in $Q$ are met; if yes, exit the loop
        \EndWhile
        
        \State
        The algorithm's output includes: (1) the release candidate dataset $\DS$; (2) the acceptance criteria results $Q$; and (3) the chosen configuration $c$.
    \end{algorithmic}
\end{algorithm}

\subsection{Leveraging Public Data}
\label{sec:intro-public-data}

Any access or computation on the Israeli National Registry of Live Births data is restricted to an enclave environment (Section~\ref{sec:components-environments}) with limited computational resources. However, the prior configuration space of possible data transformation and model hyperparameters is large and cannot be exhausted in the enclave environment in a reasonable amount of time. Not only that, if none of the configurations passes the acceptance criteria, the embedded private selection with a known threshold algorithm (Algorithm~\ref{alg:known-threshold}) in our scheme would run forever, and the privacy loss budget would be consumed for nothing.

One option is to leverage available resources (e.g., cloud computing) to run a grid search with multiple attempts for each configuration if one had public data similar to the private data. Hopefully, one or more configurations would pass all acceptance criteria with a not-too-small probability. Then, the collection of successful configuration could be used on the private data in the enclave environment.

The underlying assumption is that the structural characteristics of the data distribution would be similar between the private and public data. If a configuration, i.e., the collection of specific values for data transformations and model hyperparameters, is accepted for the public data, it would probably be accepted for the private data as well, perhaps after a few more attempts.

To the best of our knowledge, no existing birth datasets are publicly available except for those from the US, The National Vital Statistics System (NVSS) \cite{cdcnvss}, maintained by the Centers for Disease Control and Prevention (CDC).
This collection includes detailed individual-level records of birth events, categorized by state and month. The NVSS data has been accessible to the public for several decades. All six fields planned for release in the Israeli National Registry of Live Births are present in the NVSS data. Throughout the years, various privacy-protection measures were taken, such as the removal of Personal Identifiable Information (PII) and data field generalizations.

\subsection{The Release}
\label{sec:intro-release}

The execution of the our scheme on the Israeli National Registry of Live Birth resulted in a released dataset that passed all the acceptance criteria successfully. The privacy loss budget was allocated as follow: $\eps = 4$ for fitting the generative model; $\eps = 0.99$ for evaluating of the dataset with a list of predefined of acceptance criteria;  and a factor $2$ is applied due to the private selection process \cite{Liu2018PrivateSF}. Therefore, the overall privacy loss budget is $\pmb{\eps = 9.98}$. Considering past real-world differentially private releases as a benchmark \cite{desfontain2021ListRealworld}, we aimed to a total privacy loss budget of $\eps < 10$ (see Section~\ref{sec:components-differential-privacy-budget} for a discussion on setting the value of $\eps$).

Based on experiments conducted on a sample public data\footnote{Specifically, we used a sample of NVSS birth data from the year 2019 that corresponds to the approximate number of births in Israel in 2014, based on counts from a public report by Israel Central Bureau of Statistics \cite{cbs2014newborns}.} (see Section~\ref{sec:intro-public-data}), we narrowed down the space of possible configurations to those with an acceptance probability of at least 10\%. All of these configurations specify the use of the PrivBayes algorithm for training the generative model \cite{Zhang2014PrivBayesPD}. When running the scheme on the data from the Registry, a configuration featuring the data transformations outlined in Table~\ref{tab:dataset} was selected.

The stakeholders were pleased with the quality of the released dataset, as evidenced by the fact that all acceptance criteria were met and by the level of actual errors obtained, as seen in Table~\ref{tab:ac}. (For instance, in accordance with the first and the most important acceptance criterion, the released dataset demonstrates a maximal error in histograms and contingency tables of less than $0.5\%$ out of the total number of rows in the original data, smaller than the pre-set limit of $1\%$.)

To obtain a more granular quality assessment of statistical queries under the selected configuration, we generated multiple ``released datasets'' using our scheme with the public data sample as input. The analysis included queries beyond the acceptance criteria (e.g., correlations and medians). Overall, the scheme under the selected configuration produced high-quality datasets with sufficient accuracy, as confirmed by public health experts based on the granular evaluation report. 

The released dataset and its documentation tailored to different audiences \cite{moh2024release}, along with the code\footnote{\url{https://github.com/shlomihod/synthflow}} used to produce and evaluate the release, were made publicly available in February 2024.

\subsection{Summary of our Contributions}

The most immediate impact of this work is the actual release of the 2014 birth data, which, to the best of our knowledge, is the first of its kind in the medical domain. (We aware of only one other differently private release with synthetic data, in the context of human trafficking \cite{Migration2022Global}.)

More generally, this work demonstrates the feasibility of using differential privacy as a privacy protection measure for the release of government data, specifically within the medical domain.
The release represents the successful outcome of a co-design process involving various stakeholders. 
It also demonstrates how incorporating context-specific requirements regarding the stakeholders' expectations of accuracy and privacy within the overall scheme can be critical to the success of such release. We identify, conceptualize and analyze new stakeholders' requirements, namely \emph{faithfulness} and \emph{face privacy}.

On a more technical level, this work we designed a \emph{universal scheme} for producing a differentially private synthetic data that satisfies a diverse set of properties while providing an overall differential privacy guarantee. We do so by utilizing the private selection algorithm of \textcite{Liu2018PrivateSF} in a new way. Our scheme allows separating the process of designing collection of the desired properties and determining their feasibility from the actual processing of the data---that, crucially, takes place only once. The scheme is general and can be applied to other  types of differentially private releases.

\subsection{Organization of this Paper}

In the next section, we provide a detailed account of the universal scheme and design choices concerning the release (Section~\ref{sec:componenets}). Then, we review related work (Section~\ref{sec:related-work}) and discuss future research directions (Section~\ref{sec:discussion}).

%% file: sections/2-components.tex
\section{The Components of the Universal Scheme}
\label{sec:componenets}

In this section, we present the theoretical aspects and implementation details for each step of our scheme: differential privacy guarantee (Section~\ref{sec:componenets-privacy}), data preparation (Section~\ref{sec:components-data}), differential privacy-related choices such as privacy loss budget allocation (Section~\ref{sec:components-differential-privacy}), faithfulness (Section~\ref{sec:components-faithfulness}) acceptance criteria (Section~\ref{sec:components-acceptance-criteria}), dataset projection for face privacy (Section~\ref{sec:components-projection}), configuration (Section~\ref{sec:components-configurations}), constraint filtering (Section~\ref{sec:components-constraints}), and the private selection algorithm (Section~\ref{sec:components-private-selection}). 
We also elaborate on the software and environments used to execute the scheme (Sections~\ref{sec:components-software} and \ref{sec:components-environments}). We introduce the documentation that accompany the released dataset (Section~\ref{sec:components-documentation}). Finally, we conclude this section with potential improvements identified \emph{after} the scheme was executed, so they could not be integrated into the differentially private release dataset (Section~\ref{sec:components-improvements}).

\subsection{Differential Privacy Guarantee}
\label{sec:componenets-privacy}
 
We recall the notion of differential privacy, which provides a rigorous and quantifiable measure of privacy \cite{Dwork2006CalibratingNT}.
Importantly, rather than considering the privacy loss caused by a specific piece of data, the measure assesses the level of privacy loss incurred by a given \emph{mechanism} (or, algorithm) for releasing information about a dataset.
The definition considers very powerful and general adversaries, thus resulting in a strong privacy protection for each record in the dataset.  

\begin{definition}[Bounded Neighboring Datasets]
    Let $\XX$ be the universe of records. Two datasets $\DD, \mathcal{D'} \in \XX^n$ of size $n$ are \emph{neighbors} if they differ only in a single record.
\end{definition}

\begin{definition}[Differential Privacy \cite{Dwork2006CalibratingNT}]
    A randomized mechanism $\M$ is \emph{$\eps$-differently private (DP)} if for all pairs of neighboring datasets $\DD$, $\DD$' and all events Y in the output space of $\M$
    \[ \Pr[\M(\DD) \in Y] \leq e^\eps \Pr[\M(\DD') \in Y]. \]
\end{definition}

It is readily apparent that meaningful differentially private mechanisms must be randomized; in other words, they create a distribution over their outputs.
Furthermore,  there is a clear tension between the level of privacy provided by the mechanism and the accuracy of statistics on the data that are given in (or can be inferred from) the mechanism's output. This tension is commonly referred to as the \emph{privacy-accuracy trade-off}.

Two very useful properties of differential privacy mechanisms, that are central in finding an effective ``compromise'' between privacy and accuracy, are composition and post-processing:

\begin{proposition}[Differential Privacy Basic Composition \cite{DworkKMMN06}]
    Let $\{\M_i\}_{i=1}^k$ be a collection of $\{(\eps_i)\}_{i=1}^k$-differently private mechanisms, respectively. Then the their combination mechanism, $\M$, defined to be $\M(x) = (\M_1(X), \ldots, \M_k(x))$ is $\big(\sum_{i=1}^k \eps_i \big)$-differentially private.
\end{proposition}

\begin{proposition}[Differential Privacy Post-Processing \cite{Dwork2006CalibratingNT}]
    Let $\M: \XX^n \rightarrow Z$ be a $\eps$-differently private mechanism. Let $f: Z \rightarrow Z'$ be an arbitrary randomized mapping. Then $f \circ \M: \XX^n \rightarrow Z'$ is $\eps$-differently private.
\end{proposition}

Because of the composition property of differential privacy, the privacy parameter $\eps$ is also called the \emph{privacy loss budget}.

Differential privacy underpins the privacy guarantees of our scheme. All computations that take the original dataset as input must be differentially private. In our approach, this requirement extends to both the training of the generative model and the calculations of the acceptance criteria.

\subsection{Registry Data}
\label{sec:components-data}

The Live Birth Registry is the official data source of live births occurring in Israel, provided that at least one parent holds an Israeli identity card, excluding the births of Israelis that occur abroad \cite{israel23live}.

We turn to discuss two design choices regarding the preparation of the data: which year and fields to release.

\subsubsection{Releasing a single year} 

Although we have access to data from recent years, which data users generally consider more valuable, in this pilot project we decided to release data from the earliest year available to us, 2014. This choice was made for two reasons. First, it eliminates the need to release an additional column (\texttt{year}) and reduces the amount of statistical relationships that need to be learned by the generative model. Second, post-release feedback from stakeholders can be incorporated into subsequent releases of more valuable years.

\subsubsection{Selecting data fields}

The Birth Registry contains dozens of data fields for each birth. Before our collaboration with the Ministry of Health began, the primary stakeholders had already identified which data fields were planned for release, guided by the data minimization principle: selecting the smallest possible subset of fields that still delivers sufficient value to users. The pre-selected fields were also suitable for piloting with differential privacy synthetic data because they were diverse in terms of type and statistical properties (e.g., binary, categorical, continuous) and were relatively clean and complete. We considered adding one more data field (\texttt{group\_population}) but decided against it due to the high number of missing values. We will leave the release of this and other data fields for future research.

\subsection{Differential Privacy Design Choices}
\label{sec:components-differential-privacy}

Here we discuss design choices related to the differential privacy guarantees.

\subsubsection{Releasing singleton births only}

What constitutes the unit of privacy within the Birth Registry context? How should neighboring datasets be defined? At least three levels are discernible: (1) a newborn in a birth; (2) a birth, whether it involves one newborn or more; and (3) a mother.

For simplicity in this pilot project, together with the primary stakeholders, we decided to release only singleton births (births involving a single newborn), thereby obliterating the distinction between levels (1) and (2). According to primary stakeholders, the release of singleton birth data already has significant value. Multiple births are infrequent ($\approx4.6\%$ in 2014, as per the Israel Central Bureau of Statistics public report \cite{cbs2014newborns}), and are typically analyzed separately from singleton births in biostatistics and medical literature.

Consequently, we set the unit of privacy to a single singleton birth.

This is sidestepping the necessity for more intricate differential privacy algorithms designed for multiple record settings (e.g., \cite{Liu2020LearningDD}). Moreover, it is plausible to presume that a mother may have at most two pregnancies within a calendar year, with other scenarios being exceedingly rare. Hence, the vast majority of mothers will have privacy protection on the basis of birth. For a minority of mothers, protection is accorded through the composition of a group of size two.

\subsubsection{Treating the total number of live births as public information}

As articulated by several stakeholders, data users anticipate that the total number of births in the released dataset aligns with the actual figure. Any discrepancy could undermine the trustworthiness of the release. Therefore, we treat the total number of live birth \emph{after applying constrain filtering on the raw data} (Section~\ref{sec:components-constraints}) as public information. By doing so, we discard the requirement to execute a counting query and allocate privacy loss budget for it.

\subsubsection{Opting for pure differential privacy}

Conveying the assurance of differential privacy to non-experts is a challenging task \cite{Cummings2021INA,Nanayakkara2023WhatAT}. The definition encapsulates the worst-case probabilistic divergence between two world states, encompassing several non-intuitive concepts such as probability, worst-case scenarios, and hypothetical worlds \cite{Wood2018DifferentialPA}. This communicational hurdle is inescapable for those intending to release data under differential privacy. 
We aspire to mitigate this challenge by favoring pure differential privacy mechanisms, characterized by a singular parameter $\eps$, for the synthetic data generator as well as the acceptance criteria. Predominantly, we employed the simplest mechanism for acceptance criteria: noise addition from the Laplace distribution.

\subsubsection{Setting the privacy loss budget \texorpdfstring{$\eps$}{epsilon}}
\label{sec:components-differential-privacy-budget}

Differential privacy theory imparts a meaningful interpretation to the protection afforded by privacy loss budget of $\eps \leq 1$ \cite{Dwork2014TheAF}.
Nonetheless, our empirical experimentation on publicly available data (Section~\ref{sec:intro-public-data}) revealed that such a value does not yield adequate data quality with our methodology. 
Beyond this range, no explicit guidelines exist for selecting $\eps$.

To surmount this, we adhered to three heuristics to balance privacy and utility.

\begin{enumerate}
    \item Opt for the lowest $\eps$ feasible for the acceptance criteria evaluation, provided the perturbation in the error metric due to differential privacy is relatively diminutive compared to the threshold.
    \item Select the lowest $\eps$ feasible for the generative model training, as long as the release-candidate dataset has not too small probability ($\geq 10\%$) to satisfy all acceptance criteria.
    \item The aggregate $\eps$ should not exceed 10, positioning it within the mid range of privacy loss budgets utilized in previous real-world deployments \cite{desfontain2021ListRealworld}. While strictly speaking a factor of $e^{10}$ within the differential privacy definition offers only a very weak theoretical guarantee against worst-case adversaries, recent studies provide empirical evidence that a differentially private mechanism with such an $\eps$ still offers strong privacy guarantees against more realistic attackers (e.g., with limited access to the mechanism and the private dataset) \cite{Stadler2020SyntheticD}. Additionally, we further limit the attacker's information by not making the differentially private trained synthesizer public and  releasing only the final dataset.
\end{enumerate}

The actual allocation of the privacy loss budget was conducted according to the heuristics applied to empirical experiments on publicly available data (Section~\ref{sec:intro-public-data}). The chosen allocation is outlined in Section~\ref{sec:intro-release}.
Table~\ref{tab:ac} presents the breakdown per criterion.

\subsection{Faithfulness}
\label{sec:components-faithfulness}

Recall that faithfulness is the requirement of having a row-level mapping between the released dataset and the transformed dataset. The following definition formalizes this notion. 

\begin{definition}[$(\alpha, \beta)$-faithfulness]
    Let $c: \XX \times \XX \rightarrow \R_{\geq 0}$ be a cost function between two records. We say that a dataset $\DS \in \XX^n$ is \emph{$(\alpha, \beta)$-faithful} with respect to a dataset $\DR \in \XX^n$ and a cost function $c$ if there exists a 1-1 matching (bijection) $\pi: \DS \rightarrow \DR$ such that
    
    \[
        \frac{1}{n} \sum_{i=1}^n \mathbbm{1} \Big[ c(s_i,r_{\pi(i)}) \leq \alpha \Big] \geq \beta.
    \]
\end{definition}

In words, there are at most $(1-\beta) n$ records in $\DS$ (released dataset) that are not 1-1 matched to similar-enough records in $\DR$ (transformed dataset). 

For a given $\alpha$, the cost function can be scaled by a factor of $1/\alpha$, so we can assume that $\alpha = 1$, and rephrase the definition with a single parameter as \textbf{$\beta$-faithfulness}. The actual matching is not very important for the purpose of faithfulness, but its existence and value of $\beta$ matter. Therefore, we care only about the maximum possible value of $\beta$, provided by the following definition.
 
\begin{definition}[Maximal-$\beta$-faithfulness]
    \label{def:max-faithfulness}
    The maximal faithfulness of a dataset $\DS \in \XX^n$ with respect to a dataset  $\DR \in \XX^n$ and a cost function $c$  is
    \[
        \beta_\mathrm{max}(\DR, \DS) = \max_{\substack{\pi \\ \mathrm{matching}}} \frac{1}{n} \sum_{i=1}^n \mathbbm{1} [c(s_i, r_{\pi(i)}) \leq 1].
    \]
\end{definition}

$\beta_\mathrm{max}$ can be found efficiently by solving the maximum-cardinality matching problem on the bipartite graph: nodes correspond to original and released records; an edge between original and released nodes exists if the cost function is smaller or equal to 1. 

\subsubsection{Cost function}
The cost function is elicited from the stakeholder. We note that the definitions of $\beta$-faithfulness and maximal-$\beta$-faithfulness do not depend on the exact values of the cost function $c$; rather, it depends only on whether the cost between two records is less than or equal to one.
 
\subsubsection{Faithfulness vs. Quality}

Similarly to data quality evaluation, faithfulness is also a property of the relationship between the transformed and released datasets.
Data quality ensures that analyzing the released dataset as a whole produces similar statistical results as performing the same analysis on the transformed dataset.
However, the faithfulness criterion considers much finer granularity and it requires that each released dataset record ``looks like'' one distinct transformed dataset record. Namely, there is a 1-to-1 matching between the records of released dataset and the records of the transformed dataset: each released record is matched with an transformed record that is sufficiently similar. 
Note that faithfulness is a property of the two datasets, and not of the process of how they were created. For example, if one of the datasets is produced via synthetic data, the notion of faithfulness is agnostic to the chosen generative model.

\subsection{Acceptance Criteria}
\label{sec:components-acceptance-criteria}

\begin{table}[t]
\centering
\caption{\label{tab:ac}List of acceptance criteria used for evaluating this release (refer to Section~\ref{sec:components-acceptance-criteria}). Total privacy loss budget of acceptance criteria's release is $\eps = 0.99$. Recall that the total number of records in the dataset is $n = \num[group-separator={,}]{165915}$.}
\begin{adjustbox}{width=\textwidth}
\begin{tabular}{@{}lllllllll@{}}
\toprule
Type & Metric & Threshold & Result & $\eps$ & Mechanism & $[L,U]$ & $\Delta$ & $\sigma$  \\
\midrule
\multirow{2}{*}{\begin{tabular}[c]{@{}l@{}}Maximal error\\ of marginals\end{tabular}}
    & Absolute all k-way  & 1\% & 0.440\% & 0.01    & Laplace  & $[0,1]$   & $0.001>$         & 0.001
    \\
    & Relative 1-way & $\times1.4$ & $\times1.284$ & 0.30    & Laplace    & $[1, 2]$ & 0.029           & 0.135   \\
\addlinespace
\multirow{3}{*}{\begin{tabular}[c]{@{}l@{}}Maximal error\\ of conditional means\end{tabular}}
    & $\overline{\texttt{parity}}$ & 0.3 live births & -0.014 & 0.01     & Laplace  & $[1, 11]$  & $0.001>$     & 0.044   \\
    & $\overline{\texttt{birth weight}}$ & 100 grams & 28.634 & 0.17    & Laplace    & $[1400, 4600]$  & 0.459     & 3.821    \\
    & $\overline{\texttt{gestation week}}$  & 1 week & 0.062 & 0.02    & Laplace    & $[28, 43]$  & $0.001>$     & 0.033    \\
\addlinespace
\multirow{2}{*}{Linear Regression} & Max coefficient error
    & 30 & 27.185 & 0.43  & Functional &            &    \\
    & Absolute prediction error   & 5 grams & 0.351  & 0.04    & Laplace   & $[1400,4600]$   & 0.039         & 1.364  \\
    \addlinespace
Faithfulness & Record-level matching error & 5\%  & 3.876\%
 & 0.01    & Laplace   & $[0,1]$   & $0.001>$         & 0.001
 \\
\bottomrule
\end{tabular}
\end{adjustbox}
\end{table}

Each data release occurs in a specific context, that encompasses, among other factors, the needs and expectations of the data users.
For a release to be successful, it must be evaluated in a way that respects its specific context \cite{arnold2020}. To achieve this, we have constructed a set of \textbf{acceptance criteria} to ensure that the release meets the data users' requirements.

\begin{definition}[Acceptance criterion]
    \label{def:ac}
    An \emph{acceptance criterion} $\AC \defeq (\EE, T)$ is a pair of an error measure between two datasets of the same size $\EE: \XX^n \times \XX^n \rightarrow \R$ and a threshold $T \in \R$. We say that an acceptance criteria is \emph{met} for datasets $\DR, \DS \in \XX^n$ if $\EE(\DR, \DS) < T$.
\end{definition}

A \emph{release-candidate dataset} $\DS$ passes an acceptance criterion if its error metric falls below the specified threshold. The metric is computed using the \emph{transformed dataset} $\DR$ (recall that $\DR$ is created by applying data transformations to the original dataset).
We accept the release-candidate dataset only if it meets all such criteria.
These acceptance criteria and their relevant design choices, including error metrics and thresholds, are publicly pre-established by stakeholders, drawing on their subject-matter expertise and the existing literature.

Notably, acceptance criteria allow us to test any property of the released dataset, including quality-related properties, as long as it can be computed in a differentially private way. 

To preserve the end-to-end differential privacy of the scheme  (Algorithm~\ref{alg:scheme}), all acceptance criteria are computed with differential privacy, as required by the private selection algorithm (Algorithm~\ref{alg:known-threshold})
We publicly release the differentially private metric results along with the associated privacy parameters (budget, mechanism, sensitivity, and variance). These results, thresholds, and noise information provide additional value for data users, allowing them to assess whether the dataset errors are adequate for their needs. This benefit is amplified when metric results, as observed in our release, are significantly lower than thresholds, indicating stronger quality guarantees. (Releasing metric results without differential privacy could leak individual information, posing an unacceptable privacy risk \cite{Ganev2023Inadequacy}.)

For the release of the Registry, we designed, together with the stakeholders, eight acceptance criteria grouped into four types (see Table~\ref{tab:ac}):\begin{enumerate}
    \item Maximal error in contingency tables  (marginals): absolute (all $k$-way) and relative ($1$-way)
    \item Maximal error in conditional means: \texttt{parity}, \texttt{birth\_weight}, \texttt{gestation\_week}
    \item Linear regression errors: coefficients ($\ell_1$) and predictions (MAE)
    \item Faithfulness
\end{enumerate}

We present the implementation details of the eight acceptance criteria organized by type (Table~\ref{tab:ac}). We start with preliminaries (Section~\ref{sec:components-acceptance-criteria-preliminaries}), then move to the marginals-related criteria (Section~\ref{sec:components-acceptance-criteria-marginals}), conditional means-related criteria (Section~\ref{sec:components-acceptance-criteria-means}), linear regression-related criteria (Section~\ref{sec:components-acceptance-criteria-regression}), and finally, the faithfulness criterion (Section~\ref{sec:components-acceptance-criteria-faithfulness}).

\subsubsection{Preliminaries}
\label{sec:components-acceptance-criteria-preliminaries}

Let $\XX$ be the universe of records. Let $d$ be the dimension of $\XX$, i.e., number of data fields. Let $\DR, \DS \in \XX^n$ be the transformed and release-candidate datasets, respectively. When the acceptance criteria are computed in our scheme, we assume that the release-candidate dataset $\DS$ is public and fixed (e.g., it was released with another differential privacy mechanism), but the transformed dataset $\DR$ is private.

Let $M_k$ be the set of all $k$-way marginals \emph{count} queries (contingency tables) on a dataset with records from $\XX$.

\begin{definition}[Clipping Function]
    The clipping function between $L$ and $U$, denoted $[\,\cdot\,]^U_L: \R \rightarrow \R$, is defined as follows:
    \[
        [\,x\,]^U_L = \max \{ \min \{ x, U \}, L \}.
    \]
\end{definition}

A common blueprint for building a differentially private version of a function that takes a private dataset as input (e.g., count or mean) is to add sufficient random noise to the result \cite{Dwork2006CalibratingNT,DworkKMMN06}. Intuitively, the noise should be in the same magnitude as the contribution of a single record to the function output. The notion of global sensitivity quantifies this idea.

\begin{definition}[Global Sensitivity]
    The \emph{global sensitivity of a function} $f: \XX^n \rightarrow \R$ is 
    \[
        \Delta_f = \max_{\mathrm{neighbors\,} \DD, \DD'} | f(\DD) - f(\DD') |
    \]
\end{definition}

\begin{theorem}
[Laplace Mechanism \cite{Dwork2006CalibratingNT}]
    Let  $f: \XX^n \rightarrow \R$ be a function with global sensitivity $\Delta$, and let $\eps>0$. Then the algorithm $\M$, that given an dataset $\DD\in \XX^n $, outputs $\M(\DD) = f(\DD) + Z$, where $Z$ is a Laplace random variable with scale parameter $\Delta/\eps$, is $\eps$-DP. 
\end{theorem}

We proceed to bound the sensitivity of an acceptance criterion. Since most of these criteria are based on taking the maximum over multiple queries, we first provide a general bound on the global sensitivity of a maximum query.

\begin{proposition}
    \label{prop:max}
    Let $(f_i: \XX^n \rightarrow \R)_{i=1}^k$ be a collection of functions with global sensitivity values $(\Delta_i)_{i=1}^k$, respectively. Then the global sensitivity of the maximum function $f_{\max}(\DD) = \max_{i} f_i(\DD)$ is upper bounded by $\Delta_{\max} = \max_{i} \Delta_i$.

\begin{proof}
 Let $\DD_1$, $\DD_2$ be neighbor datasets. Let $f_i$ and $f_j$ be the functions that achieve the maximum value on  $\DD_1$, $\DD_2$, respectively (if more than one function achieve it, choose arbitrary). In other words, $f_{\max}(\DD_1) = f_i(\DD_1)$ and $f_{\max}(\DD_2) = f_j(\DD_2)$. Without loss of generality $f_{\max}(\DD_1) \geq f_{\max}(\DD_2)$. Then

 \begin{align*}
    |f_{\max}(\DD_1) - f_{\max}(\DD_2)| 
    = |f_i(\DD_1) - f_j(\DD_2)| 
    \leq |f_i(\DD_1) - f_i(\DD_2)|
    \leq \Delta_i \leq \Delta_{\max}
 \end{align*}
    
 We get the first inequality because $f_j$ achieve the maximal value on $\DD_2$, and the second inequality---by the definition of the global sensitivity $\Delta_i$ of a function $f_i$. We finish the proof by taking the maximum of the global sensitivities over all $i \in [n]$, which is exactly $\Delta_{\max} = \max_{i} \Delta_i$.
 
\end{proof}
\end{proposition}

In the following subsections, we present the acceptance criteria. Specifically, for each criterion, we define the error measure $\EE$, the threshold $T$, and how to compute the error measure with differential privacy.

\subsubsection{Maximal error of marginals}
\label{sec:components-acceptance-criteria-marginals}

We consider two criteria for capturing the maximal error of marginals. The first is the \textbf{maximal absolute error} of all $k$-way marginal counts, with the following error measure

\[
    \EE_{\text{abs}}(\DR, \DS) \defeq \frac{1}{n} \max_{k \in [d]} \max_{q \in M_k} | q(\DR) - q(\DS) |
\]

and threshold $T_{\text{abs}} = 0.01$ as set by the stakeholders.

It is anticipated that counts and frequencies will be the most valuable queries to the users, so this is the most important acceptance criterion. This measure is a maximum over all $q \in M_k$ for all $k \in [d]$. Each query $q$ has low sensitivity of $1$, so the overall global sensitivity is $1/n$. The error measure is released using the Laplace mechanism.

The second acceptance criterion is the \textbf{maximal relative error} of 1-marginal counts with smoothing, with the following error measure

\[
    \EE_{\text{rel}}(\DR, \DS) \defeq \max_{q \in M_1} \max \Bigg\{ \frac{\hat{q}(\DR)}{\hat{q}(\DS)}, \frac{\hat{q}(\DS)}{\hat{q}(\DR)} \Bigg\}
\]

where $\hat{q}(\DD) \defeq q(\DD) + 1$ is the smoothed count query of a count query $q$. The stakeholders set the threshold to be $T_\text{rel} = 2$.

However, as we will see (Proposition~\ref{prop:gs-1way}), $\EE_{\text{rel}}$ has high global sensitivity. 
To mitigate this, we use a common technique for reducing the sensitivity of a function: We clip its value to a bounded interval. Specifically, we modify the definition of the error measure $\EE_{\text{rel}}$ to include clipping of each component into the interval $[1, \lambda]$ for $\lambda > 1$ as follows

\[
\EE^{\lambda}_{\text{rel}}(\DR, \DS) \defeq \max_{q \in M_1} \max \Bigg\{ \Big[\frac{\hat{q}(\DR)}{\hat{q}(\DS)} \Big]^\lambda_{1}, \Big[\frac{\hat{q}(\DS)}{\hat{q}(\DR)} \Big]^\lambda_{1} \Bigg\},
\]
where $[\,\cdot\,]^U_L$ is the clipping function. Given that the two ratios are reciprocals of each other, to find the maximum the relative error, it is sufficient to confine them within the interval $[1, \lambda]$ and not $[1/\lambda, \lambda]$. 

We also adopt a threshold of $T_\text{rel} = 1.4$, which is slightly more conservative than what was set by the stakeholders. This adjustment, calculated based on the Laplace distribution CDF, aims to decrease the probability of a false positive to 5\%, where the criterion is incorrectly deemed met due to the noise added for differential privacy.

Prior to discussing the sensitivity analysis of the two variations of the error measure, we briefly present its significance in terms of data quality guarantees. Assuming ``nice'' univariate distributions (for instance, those characterized by a single mode and diminishing tails), and in conjunction with the first criterion, the 1-way maximal relative error acceptance criterion essentially imposes a bound on the error in the tails of these distributions between the transformed and the release-candidate datasets.

\paragraph{Unclipped flavor $\EE_{\text{rel}}$ has high global sensitivity.}

\begin{proposition}
    \label{prop:gs-1way}
    Assuming $\DS$ is public and fixed, the global sensitivity of $\EE_{\text{rel}}(\DR, \DS)$ with respect to $\DR$ is $\frac{s_{\max} + 1}{2}$ where $s_{\max} = \max_{q \in M_1} q(\DS)$.

\begin{proof}
    Based on Proposition~\ref{prop:max} it is sufficient to find the maximal global sensitivity of
    \[
        q'(\DR, \DS) = \frac{\hat{q}(\DR)}{\hat{q}(\DS)}
    \]
    and
    \[
        qq''(\DR, \DS) = \frac{\hat{q}(\DS)}{\hat{q}(\DR)}
    \]
    for all $q \in M_1$.
    
    First, we observe that $q'$ is an affine transformation of the count query $q(\DR)$, so $q'$ has low global sensitively equals to
    \[
        \max_{q \in M_1} 1/\hat{q}(\DS) = \frac{1}{s_{\min} + 1}
    \]
    where $s_{\min} = \min_{q \in M_1} q(\DS)$.
    
    Second, for any $\hat{q}(\DS)$, the sensitivity of $q''$ is maximized when $\hat{q}(\DR)$ is minimal, i.e.  $\hat{q}(\DR) = 1$, because
    
    \begin{align*}
         \Delta_{q''} &= \max_{\hat{q}(\DS), \hat{q}(\DR)} \Big | \frac{\hat{q}(\DS)}{\hat{q}(\DR)} - \frac{\hat{q}(\DS)}{\hat{q}(\DR) + 1} \Big| \\
         &= \Bigg( \max_{\hat{q}(\DS)} \hat{q}(\DS) \Bigg) \Bigg(\max_{\hat{q}(\DR)} \Big | \frac{1}{\hat{q}(\DR)} - \frac{\hat{q}(\DS)}{\hat{q}(\DR) + 1} \Big| \Bigg) \\
         &= (s_{\max} + 1) \cdot \Big | \frac{1}{1} - \frac{1}{1 + 1} \Big|  \\
         &= \frac{s_{\max} + 1}{2},
    \end{align*}
    
    where $s_{\max} = \max_{q \in M_1} q(\DS)$.

    Taking the maximum between $q$ and $q''$, we get that the global sensitivity is $\frac{s_{\max} + 1}{2}$.
\end{proof}
\end{proposition}

The proposition shows that the global sensitivity of $\EE_{\text{rel}}$ is of the same magnitude as the criterion's values, i.e., $\Theta(n)$. Therefore, adding Laplace noise would wipe out any signal and render the result useless.

\paragraph{Clipped-flavor $\EE^{\lambda}_{\text{rel}.}$ has low global sensitivity.}
Recall that we apply clipping to each term of $\EE^{\lambda}_{\text{rel}}$ into the interval $[1, \lambda]$ for $\lambda > 1$.

\[
\EE^{\lambda}_{\text{rel}}(\DR, \DS) \defeq \max_{q \in M_1} \max \Bigg\{ \Big[\frac{\hat{q}(\DR)}{\hat{q}(\DS)} \Big]^\lambda_{1}, \Big[\frac{\hat{q}(\DS)}{\hat{q}(\DR)} \Big]^\lambda_{1} \Bigg\}
\].

Recall also that a release-candidate dataset passes an acceptance criterion only if its result is smaller than the threshold. Because this acceptance criterion is designed to control the 1-way marginal tails, we are interested only in lower values of this ratio, which are smaller than the threshold $T_{\text{rel}} = 2$. Therefore the loss of information due to clipping is insignificant for this purpose as long as $\lambda \geq T_{\text{rel}}$.

The following proposition shows that $\EE^{\lambda}_{\text{rel}}$ has much lower sensitivity than $\EE_{\text{rel}}$.

 \begin{proposition}
    \label{prop:1way-gs-clip}
     Fix $\lambda > 1 + \frac{1}{s_{\max}}$. The global sensitivity of the acceptance criterion for the $\lambda$-clipped maximal relative error of 1-marginal frequencies $\EE^{\lambda}_{\text{rel}}$ is $\Delta = \max \Big\{ \frac{1}{s_{\min} + 1}, \lambda - \frac{1}{\frac{1}{\lambda} + \frac{1}{s_{\min} + 1}} \Big\}$, where $s_{\min} = \min{q \in M_1} q(\DS)$.

 \begin{proof}

     Based on Proposition~\ref{prop:max} it is sufficient to find the maximal global sensitivity of
    \[
        q'(\DR, \DS) = \Big[ \frac{\hat{q}(\DR)}{\hat{q}(\DS)}\Big]^\lambda_{1}
    \]
    and
    \[
        qq''(\DR, \DS) = \Big[ \frac{\hat{q}(\DS)}{\hat{q}(\DR)} \Big]^\lambda_{1}.
    \]

    Let $\DR_1, \DR_2$ be two neighbor datasets. Assume that they differ by one on the count query $q$, without loss of generality $\hat{q}(\DR_2) = \hat{q}(\DR_1) + 1$. Recall that in the evaluation step we consider $\DS$ fixed and public. Denote $s = \hat{q}(\DS)$ and $r = \hat{q}(\DR_1)$; therefore, $\hat{q}(\DR_2) = r + 1$.

    We define the functions $f$ and $g$ given $r$ and $s$ corresponding to the sensitivity of $q'$ and $q''$, respectively:

    \[
    f(r, s) = \Big[ f'(r, s) \Big]^\lambda_{1} - \Big[ f''(r, s)\Big]^\lambda_{1}\text{ where } f'(r, s)=  \frac{r + 1}{s}  \text{ and } f''(r, s) = \frac{r}{s}
    \]

and

    \[
    g(r, s) = \Big[ g'(r, s) \Big]^\lambda_{1} - \Big[ g''(r, s) \Big]^\lambda_{1}\text{ where } g'(r, s)= \frac{s}{r}  \text{ and } g''(r, s) =\frac{s}{r + 1} 
    \]

    Observe that in both cases, the first term ($f', g'$) is always larger or equal to the second term ($f'', g''$) for $s, r \geq 1$, so functions $f$ and $g$ are always non-negative in that domain. Moreover, $g$ is monotonously decreasing in $r$.
    
    Our goal is to find the supremum values of these functions which is equal to the global sensitivity of $\EE^{\lambda}_{\text{rel}}$.

    \subparagraph{Maximizing $f$.} We have five cases to consider based on whether $f'$ and $f''$ hit the upper or lower clipping boundaries. Recall that $f'(r, s) \geq f''(r, s)$.

    \begin{enumerate}
        \item If $f'(r, s), f''(r, s) \in [1, \lambda]$, then we are in the same setting of Proposition~\ref{prop:gs-1way}, and $\max f(r, s) = \frac{1}{s_{\min} +1}$.
        
        \item If $f'(r, s) = f''(r, s)$, and they are both $> \lambda$ or  $< 1$, then $f(r, s) = 0$.
        
        \item If $f'(r, s) > \lambda$, but $f''(r, s) \in [1, \lambda]$, so $\frac{r+1}{s} > \lambda$ and $\frac{r}{s} \leq \lambda$. Consequently, $r \in (\lambda s - 1, \lambda s]$, so $\sup_{r,s} f(r, s) = \sup_{r,s} \{ \lambda - \frac{r}{s} \}= \sup_{s} \{\lambda - \frac{\lambda s - 1}{s} \}= \lambda - \lambda + \sup_{s} \frac{1}{s} = \frac{1}{s_{\min} + 1}$.

        \item If $f'(r, s) \in [1, \lambda]$, but $f''(r, s) < 1$, so $\frac{r+1}{s} \geq 1$ and $\frac{r}{s} < 1$. Consequently, $r = s - 1$, so $\max_{r,s} f(r, s) = \max_{r,s} \{\frac{r + 1}{s} - 1 \} = \max_{s} \{\frac{s-1+1}{s} - \frac{1}{\lambda} \} = 1 - 1 = 0$.

        \item If $f'(r, s) > \lambda$ and $f''(r, s) < 1$. , so $\frac{r+1}{s} > \lambda$ and $\frac{r}{s} < 1$. Consequently, $\lambda s - 1 < r < s$, so $\lambda < 1 + 1/s$, in contradiction to the requirement on $\lambda > 1 + 1/s_{\max}$. Hence this case is invalid.
    \end{enumerate}

    Therefore, we conclude that the $\Delta_{q'} = \max_{r,s} f(r, s) \leq \frac{1}{s_{\min} + 1}$ for all cases, and this is the global sensitivity of $q'$.
    \subparagraph{Maximizing $g$.} We have five cases to consider based on whether $g'$ and $g''$ hit the upper or lower clipping boundaries. Recall that $g'(r, s) \geq g''(r, s)$.

    \begin{enumerate}
        \item If $g'(r, s), g''(r, s) \in [1, \lambda]$, then $\frac{s}{r} \in [1, \lambda]$ and $\frac{s}{r+1} \in [1, \lambda]$, so we have the constrain $r \in [s/\lambda, s - 1]$. $g$ is monotonically decreasing in $r$, so it is maximized with $r = s/\lambda$. Therefore, $\max_{r,s} g(r, s) = \max_s \big\{ \frac{s}{s/\lambda} - \frac{s}{s/\lambda + 1} \big\} = \max_s \big\{ \lambda - \frac{1}{1/s + 1/\lambda} \big\} = \lambda - \frac{1}{\frac{1}{s_{\min} + 1} + \frac{1}{\lambda}}$.

        \item If $g'(r, s) = g''(r, s)$, and they are both $> \lambda$ or  $< 1$, then $g(r, s) = 0$.

        \item If $g'(r, s) > \lambda$, but $g''(r, s) \in [1, \lambda]$, so $\frac{s}{r} > \lambda$ and $\frac{s}{r + 1} \leq \lambda$. Consequently, $r \in [s/\lambda - 1, s/\lambda )$. We want $\max_{r,s} g(r, s) = \lambda - \min_{r, s} g''(r, s)$. For each $s$, $g''(r, s)$ is monotonously decreasing, so it is minimized with $r = s / \lambda$, which gives us the same expression as in (1).
        
        \item If $g'(r, s) \in [1, \lambda]$, but $g''(r, s) < 1$, so $\frac{s}{r} \geq 1$ and $\frac{s}{r+1} < 1$. Consequently, $r = s$, so $\max_{r,s} g(r, s) = \max_{r,s} \frac{s}{r} - 1 = \max_{s} \frac{s+1}{s} - 1 = 1 + \max_{s} \frac{1}{s} - 1 = \max_{s} \frac{1}{s} = \frac{1}{s_{\min} + 1}$.

        \item If $g'(r, s) > \lambda$ and $g''(r, s) < 1$. , so $\frac{s}{r} > \lambda$ and $\frac{s}{r + 1} < 1$. Consequently, $s - 1 < r < s/\lambda$, so $\lambda < 1 + \frac{1}{s-1}$, in contradiction to the requirement on $\lambda > 1 + 1/s_{\max}$. Hence this case is invalid.
    \end{enumerate}

    Therefore, we conclude that the $\Delta_{q''} = \max_{r,s} g(r, s) = \Big\{ \lambda - \frac{1}{\frac{1}{s_{\min} + 1} + \frac{1}{\lambda}}, \frac{1}{s_{\min} + 1} \Big\}$ for all cases, and this is the global sensitivity of $q''$.

    In conclusion, we showed that the global sensitivity of $\EE^{\lambda}_{\text{rel}}$ is $\Delta = \max \{\Delta_{q'}, \Delta_{q''} \} \leq \max \Big\{ \frac{1}{s_{\min} + 1}, \lambda - \frac{1}{\frac{1}{\lambda} + \frac{1}{s_{\min} + 1}} \Big\}$.
\end{proof}
\end{proposition}

\paragraph{Setting the $\lambda$ and $T_{\text{rel}}$ parameters.}
Recall that a release-candidate dataset pass an accepted criterion only if the result of the error measure is lower than an predefined threshold. The primary stakeholder set the threshold for this acceptance criteria at $T_{\text{rel}} = 2$. The acceptance criteria $\EE^{\lambda}_{\text{rel}}$ is released via the Laplace mechanism. How $\lambda$ should be chosen? 

Note that the global sensitivity decreases if $\lambda$ is lower. Because the release-candidate dataset does not pass the acceptance criterion if the measure value is above the threshold, we could set $\lambda = T_{\text{rel}}$. However, in that case, all unacceptable values of the criterion collapse to a single value, namely $T_{\text{rel}}$. Consider a false positive event where the acceptance criteria is not met with the true value of the error measure, it is met when the measure is calculated with the noise addition of the Laplace mechanism. If $\lambda = T_{\text{rel}}$, then the false positive event is devastating because the true measure result could be much higher than $T_{\text{rel}}$.

Let's apply Proposition~\ref{prop:1way-gs-clip}. Based on the experiments with the public data (see Section~\ref{sec:intro-public-data}), we conservatively assume that $s_{\min} = 50$. For example, if the clipping factor is $\lambda = 2.5$, the global sensitivity is $\Delta \approx 0.119$. With $\eps = .3$, the standard deviation of the additive noise is $\sigma = \sqrt{2} \Delta / \eps \approx 0.561$. This is quite large noise compared to the threshold, and there is a high probability of a false positive event. For that reason, we proposed to drop this criterion, but the primary stakeholder assured us of its importance for data users. We could spend more privacy loss budget, but $\eps = .3$ seems already quite high for this purpose. 

To reduce sensitivity, we decided to set $\lambda = 2$. In order to avoid a false positive collapse, we opted to use a more conservative threshold $T_{\text{rel}}' < \lambda = T_{\text{rel}}$. The false positive probability $p$ is the probability that an acceptance criterion result that is clipped to $\lambda$ would be shifted to a value lower than the threshold  $T_{\text{rel}}' < \lambda$ due to the additive differentially private Laplace mechanism. We can calculate the appropriate threshold $T_{\text{rel}}'$ according to the desired $p$ with the following statement.

\begin{proposition}
    Let $\lambda > 1, \eta > 0$. Let $Z = \lambda + \mathrm{Lap}(\eta)$ be a random variable and $p = \Pr[Z \leq T']$\, for $T' \leq \lambda$. Then

    \[
    T' = \lambda + \eta \ln (2p).
    \]

\begin{proof}

    According the definition of the CDF of the Laplace distribution, we have
\[
    p = \Pr[Z \leq T'] = \Pr[\lambda + \mathrm{Lap}(\frac{\Delta}{\eps}) \leq T']
    = \Pr[\mathrm{Lap}(\eta) \leq T' - \lambda]
    = \frac{1}{2} \exp\Big(\frac{T' - \lambda}{\eta}\Big).
\]
By rearranging the terms, we get the statement. 
\end{proof}
\end{proposition}

For $p = .05$, $s_{\min}=50$, $\lambda = 2$, , it follows that $\Delta < 0.077$. Set $\eps=.3$. Hence $\eta = \Delta/\eps < .257 $, and we get that $T_{\text{rel}}' > 1.4$, which is the actual threshold we have used for this acceptance criterion. Note the values of $\Delta$ and $\sigma$ in Table~\ref{tab:ac} are slightly smaller because they are calculated from $s_{\min}$ taken from the release-candidate dataset.

\subsubsection{Maximal error in conditional means}
\label{sec:components-acceptance-criteria-means}

After frequency and count queries, measures of central tendency are the second most important considerations for users, according to our discussions with the stakeholders.

Medians, as opposed to means, are the preferred statistics for analyzing binned data like the released dataset. The stakeholders also expect the inclusion of median queries in data users' analysis. However, we chose to assess the release-candidate dataset using means for the following reasons. First, the mechanisms for computing means with differential privacy are simpler than those for medians. We favor simplicity, assuming these mechanisms (noise addition-based) are easier to explain to the stakeholders and the public. Second, in ``nice'' distributions, the error in a mean query between the transformed and release-candidate datasets often serves as an upper bound for the error in medians, since medians are less affected by extreme values. Indeed, experiments on the public data supported this assertion (Section~\ref{sec:intro-public-data}).

\textbf{Preliminaries.}
A conditional mean query $q_{a|b \leftarrow v}$ over column $a$ (averaging column) by value $v$ (group-by value) in column $b$ (group-by column) calculates the mean of column $a$ at value $v$ in the grouped-by column $b$. If $b$ is assigned the value of the ``none'' column $\bot$, then $q_{a|\bot}$ represents the mean of column $a$ without any grouping. Let $C(a)$ be a set of predefined group-by columns for an averaging column $a$, ncluding the none column $\bot$, and let $V(b)$ be a set of all values of column $b$. The maximal error of column $a$’s conditional mean acceptance criterion has the following error measure: 

\[
    \EE_{\bar{a}}(\DR, \DS) \defeq \max_{b \in C(a)} \max_{v \in V(b)} | q_{a|b \leftarrow v}(\DR) - q_{a|b \leftarrow v}(\DS) |.
\]

\begin{minipage}{\textwidth}
    Three acceptance criteria of the conditional mean type are defined:
    
    \begin{table}[H]
    \centering
    \begin{tabular}{ll}
    \toprule
    Averaging column $a$                     & Group-by columns $C(a)$  \\
    \midrule
    \texttt{parity}         &  \{\texttt{mother age}, $\bot$\}        \\
    \texttt{birth weight}   &  \{\texttt{sex}, \texttt{parity}, \texttt{gestation week}, \texttt{mother age}, $\bot$\}           \\
    \texttt{gestation week} &  \{\texttt{parity}, \texttt{mother age}, $\bot$\}
    \\
    \bottomrule
    \end{tabular}
    \end{table}
\end{minipage}

\paragraph{Releasing with differential privacy.}
To release \(\EE_{\bar{a}}(\DR, \DS)\) with differential privacy via the Laplace mechanism, an upper bound on its global sensitivity is required. According to Proposition~\ref{prop:max}, the sensitivity of the error measure \(\EE_{\bar{a}}(\DR, \DS)\) is bounded by the maximum of the sensitivities of the conditional means queries \(q_{a|b \leftarrow v}\) for each \(b \in C(a)\) and each \(v \in V(b)\).
The global sensitivity of a query \(q_{a|b \leftarrow v}\) is \((U_a - L_a)/n_{b \leftarrow v}\), where \(U_a, L_a\) are the upper and lower bounds of values in column \(a\), respectively, and \(n_{b \leftarrow v}\) is the number of records in \(\DR\) having the value \(v\) in column \(b\).
While \(U_a\) and \(L_a\) are publicly known, \(n_{b \leftarrow v}\) is not. Thus, without a lower bound on \(n_{b \leftarrow v}\), releasing the error measure via the Laplace mechanism is not feasible.

Note that \(q_{a|b \leftarrow v}\) is a quotient of two simpler queries: a sum (numerator; sum of column \(a\) for records with value \(v\) in column \(b\)) and a count (denominator; number of such records). A common strategy for releasing such composite queries is to independently compute the simpler queries with differential privacy (e.g., use the Laplace mechanism to separately compute the numerator and the denominator). Then, the results are divided, achieving privacy via post-processing.
This strategy, while straightforward, has two drawbacks.
First, noise added to the count (denominator) can significantly distort the final result.
Second, given that we are interested only in the maximum value among the \(q_{a|b \leftarrow v}\) queries for \(b \in C(a)\) and \(v \in V(b)\), this method expends privacy loss budget on calculations we do not intend to release.

Therefore, we adopt OpenDP's approach\footnote{\url{https://github.com/opendp/opendp/blob/c79ef2268bdc09cf733aba08b005b241ca63b365/docs/source/examples/unknown-dataset-size.ipynb}} for computing the mean of a dataset of unknown size by applying a resize transformation (Algorithm~\ref{alg:resize}) \cite{Gaboardi2020APF}.

\begin{algorithm}
\caption{Resized Mean (based on OpenDP implementation\protect\footnotemark)}
\label{alg:resize}
\begin{algorithmic}[1]
\Require{Dataset \(X = (x_1, x_2, \ldots, x_n) \in \R^n\), resize parameter \(m\), constant value \(w \in [L_a, U_a]\), randomness \(r\).}
\State \(s \leftarrow 0\)
\If{\(m = n\)}
    \State \(s \leftarrow \sum_{i=1}^{n} x_i\)
\ElsIf{\(m > n\)}
    \State \(s \leftarrow \sum_{i=1}^{n} x_i + (m - n) \times w\)
\ElsIf{\(m < n\)}
    \State Sample \(m\) indices \(\{i_1, \ldots, i_m\}\) from \(X\) using randomness \(r\)
    \State \(s \leftarrow \sum_{j=1}^{m} x_{i_j}\)
\EndIf
\State \Return \(s / m\).
\end{algorithmic}
\end{algorithm}

\footnotetext{\url{https://github.com/opendp/opendp/blob/c79ef2268bdc09cf733aba08b005b241ca63b365/rust/src/transformations/resize/mod.rs\#L46}}

Let \(X_\DR = (x_1, \ldots, x_{n_{b \leftarrow v}})\) be the values of column \(a\) for records with value \(v\) in column \(b\). Using this notation, \(q_{a|b \leftarrow v}(\DR) = \sum_{i=1}^{n_{b \leftarrow v}} x_i / n_{b \leftarrow v}\). We set a size \(m\), independent of \(\DR\), and ``resize'' \(X_\DR\) to create a new dataset with a publicly known size \(m\). The mean of column \(a\) is then computed as detailed in Algorithm~\ref{alg:resize}, and we denote the result by \(q_{a|b \leftarrow v}^{\widehat{m}_{b \leftarrow v}}(\DR)\), where $\widehat{m}_{b \leftarrow v}$ corresponds to the resize parameter. The following proposition outlines its global sensitivity.

\begin{proposition}[Global sensitivity of Algorithm~\ref{alg:resize}]
    \label{prop:resize}
    Let $X = (x_1, x_2, \ldots, x_n) \in \R^n$ be a private dataset of unknown size $n$. Let $L$ and $U$ be the publicly known lower and upper bounds of the values in dataset $X$, respectively. Let $\mathcal{A}(X; m, w, r)$ be the output of Algorithm~\ref{alg:resize} for dataset $X$, public resize parameter $m$, public constant value $w \in [L, U]$ and randomness $r$. Then for any $m,w,r$, the global sensitivity of $\mathcal{A}(\cdot; m, w, r)$ is $(U - L) / m$.
\end{proposition}

\begin{proof}
    Let $X_1$ and $X_2$ be two neighboring datasets.
    Note that $\mathcal{A}(X; m, w, r)$ is always the mean of $m$ in the range $[L, U]$, even with sampling or imputation.
    Let $s_i$ be the sum of the $m$ chosen elements by the algorithm for dataset $X_i$.
    \begin{itemize}
        \item If $m = n$, then the sums $s_1$ and $s_2$ are taken over all the elements in the respective detests, so they differ by a single element.  Therefore $|s_1 - s_2| \leq U - L$.
        \item If $m > n$, then each sum include all the elements in the respective dataset and additional $m - n$ copies of $w$, so $|s_1 - s_2| \leq U - L$.
        \item If $m < n$, then the same indices $\{i_1, \ldots, i_m\}$ are sampled because the randomness $r$ is fixed. Therefore the sums might differ by at most one element, so $|s_1 - s_2| \leq U - L$.
    \end{itemize}
    In conclusion, $|\mathcal{A}(X_1; m, w, r) - \mathcal{A}(X_2; m, w, r)| = |s_1/m - s_2/m| \leq (U - L)/m$.
\end{proof}

Fix $m, w$, and let $\M_r(X) = \mathcal{A}(X; m, w, r) + \mathrm{Lap}\big(\frac{U-L}{m\eps}\big)$ be the $\eps$-differentially private release of $\mathcal{A}$ with the Laplace mechanism according to the global sensitivity from Proposition~\ref{prop:resize}. Ultimately, we would uniformly sample $r$ to compute $\mathcal{A}$, resulting the following mechanism: $\widehat{\M}(\cdot) = \mathpzc{M_r}(\cdot)$ where $r \sim R$.
The following proposition shows that $\widehat{\M}$ is also $\eps$-differentially private.

\begin{proposition}
    Let \(\{\mathpzc{M_r}\}_{r\in R}\) be a family of \(\eps\)-differentially private mechanisms, where \(\mathpzc{M_r}: \XX^n \rightarrow \mathcal{Y}\).  Then, the mechanism \(\widehat{\M}(\cdot) \defeq \mathpzc{M_r}(\cdot)\) where \(r \sim R\) is also an \(\eps\)-differentially private mechanism.
\end{proposition}

\begin{proof}
    The proof is straightforward: any property of $\mathpzc{M_r}$ that holds for all $r$  will hold for a random $r$ as well. To elaborate,  let \(\DD_1\) and \(\DD_2\) be two neighboring datasets.
    \begin{align*}
        \Pr[\widehat{\M}(\DD_1) = y]
        &= \sum_{r \in \mathrm{supp}(R)} \Pr[R = r]\Pr[\M_r(\DD_1) = y] \\
        &\leq e^\eps \sum_{r \in \mathrm{supp}(R)} \Pr[R = r]\Pr[\M_r(\DD_2) = y] \\
        &= e^\eps \sum_{r \in \mathrm{supp}(R)} \Pr[R = r] \Pr[\M_r(\DD_2) = y] \\
        &= e^\eps \Pr[\widehat{\M}(\DD_2) = y]
    \end{align*}
    where the inequality holds because \(\M_r(\cdot)\) is \(\eps\)-differentially private.
\end{proof}

\paragraph{Setting the resize parameter \(\widehat{m}_{b \leftarrow v}\).}
The impact of the resize transformation on the accuracy of mean estimation depends on the gap between \(\widehat{m}_{b \leftarrow v}\) and the actual size $n_{b \leftarrow v}$. When \(\widehat{m}_{b \leftarrow v}\) is smaller, the estimation error arises from sampling. If \(\widehat{m}_{b \leftarrow v}\) is larger, the imputed values introduce error. Note that these errors behave differently: the estimator for \(m < n_{b \leftarrow v}\) is unbiased, but the same cannot be said for the case \(\widehat{m}_{b \leftarrow v} > n_{b \leftarrow v}\) without additional information about column \(a\). From a quality perspective, it is generally preferable to set \(\widehat{m}_{b \leftarrow v}\) as close to the actual dataset size as possible; if an error must be made, it is better for \(\widehat{m}_{b \leftarrow v}\) to be smaller.

Thanks to the first acceptance criterion, which is based on the maximum absolute error of all \(k\)-marginal counts, we can derive a lower bound for the actual dataset size with high probability by obtaining this count from the release-candidate dataset. We denote this lower bound as \(\widehat{m}_{b \leftarrow v}\). Let $T_{\text{abs}}$ be the predefined threshold of the first acceptance criterion. If a release-candidate dataset passes the first acceptance criterion, then any \(k\)-way marginal count query has an absolute error of up to \(n \cdot T_{\text{abs}}\) when compared to the transformed dataset. Thus, we can estimate the number of records with value \(v\) in column \(b\) without spending any privacy loss budget, as the release-candidate dataset is considered public and fixed at this stage of our scheme. The chosen resize parameter for the query \(q_{a|b \leftarrow v}^{\widehat{m}_{b \leftarrow v}}(\DR)\) is then \(\widehat{m}_{b \leftarrow v} \defeq \max \{1, q_{b \leftarrow v}(\DS) - n \cdot T_{\text{abs}} \}\), where \(q_{b \leftarrow v}(\DS)\) represents the number of occurrences of value \(v\) in column \(b\) in the release-candidate dataset \(\DS\).

\paragraph{Putting it all together.}
We modify the acceptance criterion to include the resize transformation, 

\[
    \widehat{\EE}_{\bar{a}}(\DR, \DS) \defeq \max_{b \in C(a)} \max_{v \in V(b)} | q_{a|b \leftarrow v}^{\widehat{m}_{b \leftarrow v}}(\DR) - q_{a|b \leftarrow v}(\DS) |.
\]

Observe that $\widehat{\EE}_{\bar{a}}$ is a maximum taken over all $b \in C(a)$ and $v \in V(b)$. Therefore, we can apply Proposition~\ref{prop:max} to find the global sensitivity. The global sensitivity of an individual term defined by $(b, v)$ is $\Delta_{b \leftarrow v} = (U_a - L_a)/\widehat{m}_{b \leftarrow v}$ according to Proposition~\ref{prop:resize}, so

\[
    \Delta_{\max}
    = \max_{b,v} \Delta_{b \leftarrow v}
    = \max_{b,v} \frac{U_a - L_a}{\widehat{m}_{b \leftarrow v}}
    = \frac{U_a - L_a}{\min_{b,v} \widehat{m}_{b \leftarrow v}}.
\]

Note that $\min_{b,v} \widehat{m}_{b \leftarrow v}$ represents the minimal occurrence in the release-candidate dataset $\DS$ of any value in column $b$ with the additional margin. Therefore, $\widehat{\EE}_{\bar{a}}$ is released using the Laplace mechanism with global sensitivity $\Delta_{\max}$.

\paragraph{Additional Technical Details}

\subparagraph{Converting binned columns to numeric values.}

Calculating the mean of a binned column is not straightforward. We apply the following heuristic to convert a bin into a single numerical value:

\begin{enumerate}
    \item If the bin contains only a single value (e.g., 1), that value is used (1).
    \item If the bin has defined bounds (e.g., 2-3), the average of the boundary values is used (2.5).
    \item If the bin is unbounded (e.g., >10), the specified edge value is used (10).
\end{enumerate}

\subparagraph{Coarse binning for group-by columns.}

Rather than using the group-by columns as given, we transform them to create coarser bins.  They were chosen to mimic anticipated group-by queries by data users. The bins were created by the stakeholders based on their expertise and the literature. 

\begin{enumerate}
    \item \texttt{mother\_age}: $\leq$24, 25--29, 30--34, 35$\leq$
    \item \texttt{parity}: 1, 2--3, 4$\leq$
    \item \texttt{gestation\_week}: $<$37, 37$\leq$
    \item \texttt{sex}: M, F
    \item \texttt{birth\_weight}: $<$2500, 2500--3999, 4000$\leq$
\end{enumerate}

\subsubsection{Linear regression}
\label{sec:components-acceptance-criteria-regression}

Linear regressions allow us to succinctly describe the relationship between multiple variables by assuming a simple structure. Even though linear regressions may not capture complex dependencies as multi-way marginals do, they can provide deeper insights by quantifying the effect of each feature variable on the target variable. When appropriate, linear regressions have the advantage of accounting for the distance between values of a column (e.g., parity 1 is closer to 2 than to 4), in contrast to multi-way marginal queries. 

Fitting a linear regression model is a typical analysis done on birth data, particularly with \texttt{birth\_weight} as the target variable. Linear regression can be used to serve two purposes: description and prediction \cite{Shmueli2011ToEO}. In descriptive modeling, an analyst aims to summarize data, and in the context of linear regression, this corresponds to finding the model's coefficients. Prediction modeling aims to propose a value for the target variable given the rest of the variables (features) in a new record. We defined two acceptance criteria in accordance with these two purposes.

\paragraph{Preliminaries.}
Let $\DD$ be a dataset of $k + 1$ columns. Without loss of generality, the dataset may or may not have a column of all 1s (corresponding to the constant term in the regression). Let $w_c(\DD) \in \R^{k}$ be the coefficient vector of a linear regression trained on $\DD$ to predict the last column $c$ based on the rest of the columns using standard OLS optimization. If the same linear regression model is trained with $\eps$-differentially private Functional Mechanism \cite{Zhang2012FunctionalMR}, we denote the model coefficient vector as $w_c^\eps(\DD) \in \R^{k}$.

The binned data is transformed into real numbers with the transformation used for the conditional means acceptance criteria (Section~\ref{sec:components-acceptance-criteria-means}). To make the error in the coefficients' entries comparable to each other, they are standardized using the mean and variance calculated from the release-candidate dataset $\DS$, so no privacy loss budget is consumed.

Recall that $[\,x\,]^U_L \defeq \max \{ \min \{ x, U, \}, L \}$ is the clipping function within the boundaries $L < U$. The (clipped) Maximum Average Error (MAE) of a linear regression with a coefficient vector $w$ on a dataset $\DD$ with $L$ and $U$ as the boundaries of the target variable $y \in [L, U]$ is
\[
    \mathrm{MAE}_\DD(w)
    = \frac{1}{n} \sum_{(x, y)\in \DD} \big|y - [w^\intercal \cdot x]^U_L \big|
\]

Together with the primary stakeholders, we chose $c = \texttt{birth\_weight}$ for the target column, i.e., the column $y$ that the linear regression predicts.

\paragraph{Descriptive modeling: coefficient error.}
For descriptive modeling, the acceptance criterion of the \textbf{max coefficient error} of a linear regression has the following error measure

\[
    \EE_{\text{lr-coef}}(\DR, \DS) \defeq \lVert w_c^\eps(\DR) - w_c(\DS) \rVert_1,
\]

with threshold $T_{\text{lr-coef}} = 30$.

We chose the $\ell_1$ distance because it bounds both the maximum and total errors, provided the threshold is a relatively low value.

After producing the released data, we observed that we could save $\eps = .43$ privacy loss budget spent on this acceptance criteria. Refer to Section~\ref{sec:components-improvements-lr} for details.

\paragraph{Predictive modeling: absolute prediction error.}
For predictive modeling, we follow the ``Train on Synthetic, Test on Real'' (TSTR) approach to evaluate the quality of predictions \cite{Esteban2017RealvaluedT}. We assess whether training a linear regression on the synthetic data produces accurate predictions on the \emph{transformed data}. Our baseline is a linear regression trained directly on the transformed data.

The acceptance criterion of the \textbf{absolute prediction error} of a linear regression has the following error measure
\[
    \EE_{\text{lr-mae}}(\DR, \DS) \defeq | \text{MAE}_\DR(w_c^\eps(\DR)) - \text{MAE}_\DR(w_c(\DS)) |
\]

with a threshold of $T_{\text{lr-mae}} = 5$.

To release this metric with differential privacy, we consider the coefficient vector of the linear regression $w_{\DR} \defeq w_c^\eps(\DR)$ as public since it is released in the previous acceptance criterion. Similarly, $w_{\DS} \defeq w_c(\DS)$ is public because it is computed from the release-candidate dataset $\DS$ produced from a differentially private generative model. As shown in the proposition below, the metric $\EE_{\text{lr-mae}}(\DR, \DS)$ has relatively small sensitivity, so it is released using the Laplace mechanism.

\begin{proposition}
    Let $w_{\DR}$ and $w_{\DS}$ be two public linear regression coefficient vectors. Let $L$ and $U$ be the lower and upper bounds of the target variable. The global sensitivity of the error measure \emph{$\EE_{\text{lr-mae}}(\DR, \DS)$} is $2(U-L)/n$.

\begin{proof}
    Let $\DR_1$, $\DR_2$ be neighbor datasets differ in one element, without loss of generality, the last one $(x_{1,n}, y_{1,n}) \neq (x_{2,n}, y_{2,n})$. The global sensitivity is bounded as follows:
    \begin{align*}
        &\Big| \EE_{\text{lr-mae}}(\DR_1, \DS)
        - \EE_{\text{lr-mae}}(\DR_2, \DS) \Big| \\
        &= \Big| \big| MAE_{\DR_1}(w_\DR) - MAE_{\DR_1}(w_\DS) \big|
        - \big| MAE_{\DR_2}(w_\DR) - MAE_{\DR_2}(w_\DS) \big| \Big| \\
        &\leq \Big| MAE_{\DR_1}(w_\DR) - MAE_{\DR_1}(w_\DS)
        - MAE_{\DR_2}(w_\DR) + MAE_{\DR_2}(w_\DS) \Big| \\
        &= \Big| \frac{1}{n} \sum_{i=1}^n \big|y_{1,i} - [w_{\DR}^\intercal \cdot x_{1,i}]^U_L \big| - \frac{1}{n} \sum_{i=1}^n \big|y_{1,i} - [w_{\DS}^\intercal \cdot x_{1,i}]^U_L \big| \\
        &-\frac{1}{n} \sum_{i=1}^n \big|y_{2,i} - [w_{\DR}^\intercal \cdot x_{2,i}]^U_L \big| + \frac{1}{n} \sum_{i=1}^n \big|y_{2,i} - [w_{\DS}^\intercal \cdot x_{2,i}]^U_L \big| \Big|\\
        &= \frac{1}{n} \Big| \big|y_{1,n} - [w_{\DR}^\intercal \cdot x_{1,n}]^U_L \big| - \big|y_{1,n} - [w_{\DS}^\intercal \cdot x_{1,n}]^U_L \big| \\
        &- \big|y_{2,n} - [w_{\DR}^\intercal \cdot x_{2,n}]^U_L \big| + \big|y_{2,n} - [w_{\DS}^\intercal \cdot x_{2,n}]^U_L \big| \Big|\\
        &\leq \frac{1}{n} \Big| \big|y_{1,n} - [w_{\DR}^\intercal \cdot x_{1,n}]^U_L \big| - \big|y_{1,n} - [w_{\DS}^\intercal \cdot x_{1,n}]^U_L   \big| \Big| \\
        &+ \frac{1}{n} \Big| - \big|y_{2,n} - [w_{\DR}^\intercal \cdot x_{2,n}]^U_L \big| + \big|y_{2,n} - [w_{\DS}^\intercal \cdot x_{2,n}]^U_L \big| \Big|\\
        &\leq \frac{1}{n} \Big| y_{1,n} - [w_{\DR}^\intercal \cdot x_{1,n}]^U_L - y_{1,n} + [w_{\DS}^\intercal \cdot x_{1,n}]^U_L \Big| \\
        &+ \frac{1}{n} \Big| - y_{2,n} + [w_{\DR}^\intercal \cdot x_{2,n}]^U_L + y_{2,n} - [w_{\DS}^\intercal \cdot x_{2,n}]^U_L \Big|\\
        &= \frac{1}{n} \Big| [w_{\DR}^\intercal \cdot x_{1,n}]^U_L - [w_{\DS}^\intercal \cdot x_{1,n}]^U_L \Big|
        + \frac{1}{n} \Big| [w_{\DR}^\intercal \cdot x_{2,n}]^U_L - [w_{\DS}^\intercal \cdot x_{2,n}]^U_L \Big|\\
        &\leq 2(U - L)/n.
    \end{align*}

The reverse triangle inequality gives the first and third inequalities. The triangle inequality gives the second inequality. The last inequality is given because each term is between $U$ and $L$.
\end{proof}
\end{proposition}

\subsubsection{Faithfulness}
\label{sec:components-acceptance-criteria-faithfulness}

Section~\ref{sec:components-faithfulness} offers an overview of the faithfulness acceptance criterion.
Here, we present the technical details related to the threshold and the cost function in this release. We also detail the specifics of the differentially private mechanism used for this criterion.

The faithfulness acceptance criterion has the following error measure
\[
\EE_{\text{ff}}(\DR, \DS) \defeq 1 - \beta_\mathrm{max}(\DR, \DS),
\]

where $\beta_\mathrm{max}$ stands for the maximal-$\beta$-faithfulness of a dataset $\DS$ with respect to a dataset  $\DR$ (see Definition~\ref{def:max-faithfulness}). Recall that $\beta_\mathrm{max}$ can be computed efficiently see (Section~\ref{sec:components-faithfulness}).

The stakeholders set the threshold at $T_{\text{ff}} = 0.05$. In words, no more than 5\% of the records in $\DS$ would remain unmatched to records in $\DR$.

The $\beta_\mathrm{max}$ has low sensitivity, so it is released with the Laplace mechanism.

\begin{proposition}
    The global sensitivity of the faithfulness acceptance criteria $\EE_{\text{ff}}(\DR, \DS) = 1 - \beta_\mathrm{max}(\DR, \DS) = 1 - \max_{\pi} \frac{1}{n} \sum_{i=1}^n \mathbbm{1} [c(s_i, r_{\pi(i)}) \leq 1]$ is $1/n$, where $\DS$ is public.

\begin{proof}
    Let $\DR_1$, $\DR_2$ be neighbor datasets that, without loss of generality, differ in the last element $r_{1,n} \neq r_{2,n}$. Let $\pi_1$ and $\pi_2$ be their matchings that realize the optimal value. Let $\sigma_1 = \inv{\pi_1}$, $\sigma_2 = \inv{\pi_2}$  and assume $|S| = |R_1| = |R_2|$. Without loss of generality, $\beta_\mathrm{max}(\DR_1) \geq \beta_\mathrm{max}(\DR_2)$. Then the global sensitivity is
    
    \begin{align*}
        \Big| \EE_{\text{ff}}(\DR_2) - \EE_{\text{ff}}(\DR_1) \Big|
        &= \Big| \Big(1 - \beta_\mathrm{max}(\DR_2) \Big) - \Big(1 - \beta_\mathrm{max}(\DR_1) \Big) \Big| \\
        & = \Big| \beta_\mathrm{max}(\DR_1) - \beta_\mathrm{max}(\DR_2) \Big| \\
        &= \Big| \frac{1}{n} \sum_{i=1}^n \mathbbm{1}[c(s_i, r_{1, \pi_1(i)}) \leq 1]  - \frac{1}{n} \sum_{i=1}^n  \mathbbm{1}[c(s_i, r_{2, \pi_2(i)}) \leq 1] \Big| \\
        &= \frac{1}{n} \Big| \sum_{j=1}^n \mathbbm{1} [c(s_{\sigma_1(j)}, r_{1,j}) \leq 1] - \sum_{j=1}^n  \mathbbm{1} [c(s_{\sigma_2(j)}, r_{2,j}) \leq 1] \Big| \\
        &\leq \frac{1}{n} \Big| \sum_{j=1}^n \mathbbm{1} [c(s_{\sigma_1(j)}, r_{1,j}) \leq 1] - \sum_{j=1}^n  \mathbbm{1} [c(s_{\sigma_1(j)}, r_{2,j}) \leq 1] \Big| \\
        &= \frac{1}{n} \Big| \sum_{j=1}^n \Big( \mathbbm{1} [c(s_{\sigma_1(j)}, r_{1,j}) \leq 1]  -  \mathbbm{1} [c(s_{\sigma_1(j)}, r_{2,j} \leq 1] \Big) \Big|  \\
        &= \frac{1}{n} \Big|  \mathbbm{1} [c(s_{\sigma_1(n)}, r_{1,n}) \leq 1]  -  \mathbbm{1} [c(s_{\sigma_1(n)}, r_{2,n}) \leq 1] \Big|  \\
        &\leq \frac{1}{n}.
    \end{align*}
    
    The first inequality holds because $\sigma_2$ is the maximizer of $\beta_\mathrm{max}(\DR_2)$, and the second inequality holds because $\mathbbm{1}[\,\cdot\,]$ can evaluate to either 0 or 1.

\end{proof}
\end{proposition}

Recall that the definition of maximal-$\beta$-faithfulness is independent of the exact value of the cost function
$c$. It relies solely on whether the cost between two records is less than or equal to one. The criteria under which the cost function between two records does not exceed $1$ were predefined by the stakeholder, drawing on subject-matter expertise. The conditions for a possible record match follow.

\begin{itemize}
    \item The following columns must have identical values between the two records: \texttt{birth\_month}, \texttt{parity}, \texttt{birth\_sex}.
    \item These columns can differ by one bin either up or down, but only for one column: \texttt{mother\_age}\footnote{Based on input from the primary stakeholder, there is a unique exception to this rule. If the value 37, which is clinically meaningful, falls within a bin and not at its edges, that bin must be an exact match. This is only relevant for one binning alternative of the \texttt{mother\_age} column (see Section~\ref{sec:components-configurations}).}, \texttt{gestation\_week}, \texttt{birth\_weight}.
\end{itemize}

\subsection{Face Privacy and Dataset Projection}
\label{sec:components-projection}

To satisfy the expectations regarding a privacy-protected release, face privacy, we introduce the following data projection. Algorithm~\ref{alg:projection} presents the projection for minimal occurrence, i.e., it outputs a dataset in which each record appears at least \texttt{min\_count} times. 

Let $\DD$ be a multiset (dataset) $\DD: \XX \rightarrow \mathbb{Z}_{0+}$, and let $|\DD| \defeq \sum_{x \in \XX} \DD(x)$ be the size of dataset $\DD$. Let $\REC(\DD)$ be the set of all records that appear at least once in $\DD$, i.e., $\REC(\DD) \defeq \{x \in \XX | \DD(x) > 0 \}$. Define also the set $\REC(\DD,\#=k) \defeq \{ x \in \XX: \DD(x) = k \}$ and denote its size as $n_k^\DD \defeq |\REC(\DD,\#=k)|$.

\begin{algorithm}
\caption{Dataset Projection for Minimal Occurrence}\label{alg:projection}
\begin{algorithmic}[1]
    \Require a multiset (dataset) $\DD$, a minimum count of records $m$.
    \For{$k = 1, \ldots, m-1$}
        \State $R_k \leftarrow$ Sample without without replacement $\floor{\frac{k}{m} n^\DD_k}$ records from $\REC(\DD, \#=k)$
        \State $\DD'_{k} \leftarrow \{ (x, m): x \in R_k \}$
    \EndFor
    \State $\DD'_{< m} \leftarrow \cup_{k=1}^{m-1} \DD'_{k}$
    \State $\DD'_{=m} \leftarrow \{(x, \DD(x)): x \in \XX, \DD(x) = m \}$
    \State $\DD'_{> m} \leftarrow \{(x, \DD(x)): x \in \XX, \DD(x) > m \}$
    \State $\DD' \leftarrow \DD'_{< m} \cup \DD'_{=m} \cup \DD_{> m}'$
    \State \Return $\DD'$.
\end{algorithmic}
\end{algorithm}

Algorithm~\ref{alg:projection} produces an output dataset whose records are a subset of the input dataset. All records in the output dataset appear at least $m$ times, and the total count of the records appearing $k \leq m$ times in the input dataset is preserved. We prove these properties in the following proposition.

\begin{proposition}
    \label{prop:projection}
    Let $\DD$ be a multiset (dataset) of size $n$ and $\DD'$ be the multiset corresponding to the output of Algorithm~\ref{alg:projection} when running it on $\DD$ with a minimum count of records $m$ (\texttt{min\_count}). Assume $\forall k \in [m-1]: m \mid (k \cdot n^\DD_k)$. Then the following holds:
    \begin{enumerate}
        \item $\REC(\DD') \subseteq \REC(\DD)$,
        \item For all $k > m$, $\REC(\DD',\# = k) = \REC(\DD,\# = k)$,
        \item $\REC(\DD,\# = m) \subseteq \REC(\DD',\# = m)$,
        \item $\REC(\DD',\# < m) = \varnothing$,
        \item For all $k < m$, $\sum\limits_{x \in \REC(\DD,\#=k)} \DD'(x) = \sum\limits_{x \in \REC(\DD,\#=k)} \DD(x)$,
        \item $|\DD'| = |\DD|$.
    \end{enumerate}

\begin{proof}
    Statement (1) is straightforward; the algorithm does not add a new record $x \in \XX$ that is not already present in $\DD$.

    To show that statements (2), (3) and (4) hold, note that $\DD'$ is a union of three disjoint sets: $\DD'_{>m}$, $\DD'_{=m}$ and $\DD'_{< m}$.\\
    $\DD'_{>m}$ is an exact copy of the records in $\DD$ that appears $m+1$ times or more, so for all $k > m$, $\REC(\DD',\# = k) = \REC(\DD,\# = k)$.
    Similarly, $\DD'_{=m}$ contains a copy of the records in $\DD$ that appears exactly $m$ times, so $\REC(\DD,\# = m) \subseteq \REC(\DD',\# = m)$.
    Finally, each record of $\DD'_{<m}$appears exactly $m$ times in $\DD'$. Putting it all together, there is no record in $\DD'$ that appears less than $m$ times. Hence $\REC(\DD',\# < m) = \varnothing$.

    For $k < m$, $|R_k| = \floor{\frac{k}{m} n^\DD_k} = \frac{k}{m} n^\DD_k$ because we assume that $m | (k \cdot n^\DD_k)$. Therefore, 
        \[|\DD_k'| = m \cdot |R_k| = m \cdot \frac{k}{m} n^\DD_k = k \cdot n^\DD_k.
    \]

    Consequently

    \[
    \sum\limits_{x \in \REC(\DD,\#=k)} \DD'(x) 
        = \sum\limits_{x \in R_t} \DD'(x)
        + \sum\limits_{x \in \REC(\DD,\#=k) \setminus R_t} \DD'(x) 
        = |\DD_k'|
        + \sum\limits_{x \in \REC(\DD,\#=k) \setminus R_t} 0
        = k \cdot n^\DD_k.
    \]

    Then, we derive statement (5) because $\sum\limits_{x \in \REC(\DD,\#=k)} \DD(x) = k \cdot |\REC(\DD,\#=k)| = k \cdot n^\DD_k$.
    
    Finally, statement (6) is given by
    \begin{align*}
        |\DD'|
        &= |\DD'_{>m}| + |\DD'_{= m}| + |\DD'_{<m}| \\
        &= \sum\limits_{x \in \REC(\DD,\#>m)} \DD'(x) 
        + \sum\limits_{x \in \REC(\DD,\#=m)} \DD'(x) 
        + \sum\limits_{x \in \REC(\DD,\#<m)} \DD'(x) \\
        &= \sum\limits_{x \in \REC(\DD,\#>m)} \DD(x) 
        + \sum\limits_{x \in \REC(\DD,\#=m)} \DD(x) 
        + \sum\limits_{x \in \REC(\DD,\#<m)} \DD(x) \\
        &= |\DD|.
    \end{align*}

\end{proof}
\end{proposition}

\begin{remark}
    If $\exists k \in [m-1]$ that $m \nmid (k \cdot n^\DD_k)$, then the size of the output dataset $\DD'$ would be smaller than the size of the input dataset $\DD$. To overcome that, $|\DD| - |\DD'|$ records from $\DD'$ should be duplicated. If $m$ is rather small (e.g., $m \in \{2, 3\}$, as in our release), the impact of this change is negligible.
\end{remark}

\begin{remark}
    In the actual execution of our scheme to produce the released dataset, we used a slightly different version of Algorithm~\ref{alg:projection} of which statement (5) from Proposition~\ref{prop:projection} holds in expectation. Only after producing the released dataset, which consumed the privacy loss budget, we realized that the implemented algorithm could be improved to Algorithm~\ref{alg:projection}. Refer to Section~\ref{sec:components-improvements-projection} for additional information.
\end{remark}

The parameter \texttt{min\_count} is part of the configuration space, and needs to be tuned (see Section~\ref{sec:components-configurations}).

\subsection{Configuration Space}
\label{sec:components-configurations}

A configuration is the complete specification requires for producing a release-candidate dataset from the original dataset. It consists information about the (1) data transformation; (2) generative model family; (3) its hyperparameters; and (4) dataset projections. 
In this work, we performed a search within the Cartesian space spanned by these four dimensions, as detailed in the following.

\subsubsection{Data transformations}

In many use cases, there might be more than one option acceptable by the stakeholders regarding the released dataset schema, i.e., how each data column is represented. For example, a date column could be represented at the resolution of months or weeks, and both could carry value to the users. Another example, should we treat the column \texttt{birth\_weight} as a continuous or categorical variable?

These design choices might yield different quality outcomes; for example, if the privacy loss budget is fixed, finer binning enables finer analysis, but it might also be less accurate because relatively more noise is added due to bins with smaller counts.
By eliciting a space of possible per-column data transformations and representations from the stakeholders, we avoid committing to a single combination before producing the released dataset, allowing for greater flexibility in generating a high-quality dataset.
The data transformations also interact with the model family and its hyperparameters in a way that might not be predictable ahead of time. To the best of our knowledge, this is the first instance of considering multiple alternatives for data transformations as an integral part of a differentially private release.

For each column there is one or more data transformations alternatives.

\begin{table}[H]
\centering
\begin{tabular}{ll}
\toprule
Column                           & Transformation Alternatives                                                          \\
\midrule
\texttt{birth\_month}                     & 1, 2, 3, 4, 5, 6, 7, 8, 9, 10, 11, 12                                                \\
\addlinespace
\multirow{3}{*}{\texttt{mother\_age}}     & \textless{}18, 18-19, 20-24, 25-29, 30-34, 35-39, 40-42, 43-44, 44\textless{}        \\
                                 & \textless{}18, 18-19, 20-24, 25-29, 30-34, 35-36, 37-39, 40-42, 43-44, 44\textless{} \\
                                 & \textless{}18, 18, 19, 20, ..., 42, 43, 44, 44\textless{}                            \\
\addlinespace

\multirow{2}{*}{\texttt{parity}}          & 1, 2-3, 4-6, 7-10, 10\textless{}                                                     \\
                                 & 1, 2, 3, 4, 6, 7, 8, 9, 10, 10\textless{}                                            \\
\addlinespace

\multirow{2}{*}{\texttt{gestation\_week}} & \textless{}29, 29-31, 32-33, 34-36, 37-41, 41\textless{}                             \\
                                 & \textless{}29, 29, 31, 32, ..., 38, 39, 40, 41\textless{}                            \\
\addlinespace
\texttt{birth\_sex}                       & M, F                                                                                 \\
\addlinespace
\texttt{birth\_weight}                    & \textless{}1500, 1500-1599, 1600-1699, …, 4300-4399, 4400-4499, 4499\textless{}      \\ 
\bottomrule
\end{tabular}
\end{table}

\subsubsection{Model family}

A synthetic data mechanism learns some private representation of the data distribution and outputs a generative model. We chose to work with three families of differentially private synthetic data mechanisms: (1) Marginals-based (PrivBayes) \cite{Zhang2014PrivBayesPD}; (2) Query-based (MWEM, PEP) \cite{Hardt2010ASA,Liu2021IterativeMF} and (3) Deep Learning (DPCTGAN, PATECTGAN) \cite{Rosenblatt2020DifferentiallyPS,Xu2019ModelingTD,Jordon2019PATEGAN}. Generally speaking, private marginals-based mechanisms seem to perform well in practice in settings similar to ours \cite{Tao2021BenchmarkingDP,ganev2023understanding}, query-based mechanisms tend to have theoretical utility guarantees, and deep learning approach are more common in non-differentially private settings \cite{Jordon2022SyntheticD,qian2023Synthcity}.

Within each family, the specific mechanisms were included in the configuration space only if they were available in the well-maintained open-source package \emph{SmartNoise} \cite{ms2020smartnoise}. There were two exceptions to this rule: PrivBayes and PEP. PrivBayes is widely regarded as a solid baseline for differentially private synthetic data generation, with its original implementation in C extensively used across various packages \cite{Zhang2014PrivBayesPD}. PEP is a more sophisticated query-based mechanism that outperforms MWEM and shows better empirical results than other query-based mechanisms in a similar setting to ours \cite{Liu2021IterativeMF}. Other potential candidates for a marginal-based mechanism include the MST \cite{McKenna2021WinningTN} or AIM \cite{McKenna2022AIMAA} algorithms, which have demonstrated top performances in recent benchmark analyses \cite{Tao2021BenchmarkingDP,McKenna2022AIMAA,McKenna2021WinningTN}; however, they were only added to SmartNoise at a later stage of this project.

\subsubsection{Hyperparameter space}

\paragraph{PrivBayes.}

PrivBayes \cite{Zhang2014PrivBayesPD} is a differentially private Bayesian Network based mechanism. Some portion of the privacy loss budget is allocated to learning the network structure with a greedy algorithm, and the other portion is devoted to learning the conditional distribution for each node.

The code is taken from the original implementation with a few modifications. We added a new hyperparameter (\texttt{epsilon\_split}) that determined how to divide the privacy loss budget between structure and distribution learning. The original PrivBayes mechanism has a single hyperparameter, \texttt{theta}, that heuristically tunes each node's degree (number of dependencies). We created a second flavor with a hyperparameter, \texttt{degree}, that set the maximum degree directly.

\begin{table}[H]
\centering
\begin{tabular}{lll}
\toprule
                   & Hyperparameters & Possible Values                              \\
\midrule
All models         & \texttt{epsilon\_split}   & 0.1, 0.25, 0.5, 0.7                          \\
\addlinespace
Theta flavor only  & \texttt{theta}           & 2, 4, 8, 16, 20, 25, 30, 35, 40, 50, 60, 100 \\
\addlinespace
Degree flavor only & \texttt{degree}          & 2, 3, 4                                 
\\
\bottomrule
\end{tabular}
\end{table}

\paragraph{MWEM \& PEP.}

MWEM \cite{Hardt2010ASA} and PEP \cite{Liu2021IterativeMF} are query-based mechanisms that learn a representation of the universe of records distribution induced by the transformed data. First, a pool of counting queries is randomly generated (\texttt{num\_query}). Second, for multiple iterations (\texttt{num\_iteration}), the most poorly performing counting query on the representation is selected and evaluated in a differentially private way. Then, this query result updates the representation multiple times (\texttt{num\_inner\_updates}). The code of PEP has an additional hyperparameter, \texttt{marginal}, that sets the order of the marginal queries.

Hyperparameter combinations where $\texttt{num\_query} < \texttt{num\_iterations}$ are excluded.

\begin{table}[H]
\centering
\begin{tabular}{lll}
\toprule
                            & Hyperparameters     & Possible Values       \\
\midrule
\multirow{3}{*}{All models} & \texttt{num\_query}          & 128, 512, 1024 , 4096 \\
                            & \texttt{num\_iterations}     & 100, 500, 1000        \\
                            & \texttt{num\_inner\_updates} & 25, 100               \\
\addlinespace
PEP only                    & \texttt{marginal}            & 2, 3, 4      \\
\bottomrule
\end{tabular}
\end{table}

\paragraph{DPCTGAN \& PATECTGAN.}

Both DPCTGAN and PATECTGAN \cite{Rosenblatt2020DifferentiallyPS,Jordon2019PATEGAN} are variants of the CTGAN model \cite{Xie2018DPGAN}; the former uses DP-SGD for training CTGAN directly, while the latter embeds it within the PATE framework \cite{papernot2017semi,nicolas2018scalable}.  Unless otherwise specified, hyperparameters default to the values in SmartNoise.

The parameter \texttt{epochs} denotes the maximum number of training epochs, applicable even if the privacy loss budget remains unexhausted. \texttt{batch\_size} defines the size of each training batch. Learning rates and weight decays for generator and discriminator networks are configured by \texttt{generator\_lr}, \texttt{discriminator\_lr}, \texttt{generator\_decay}, and \texttt{discriminator\_decay}, respectively.

The hyperparameter \texttt{noise\_multiplier} adjusts the amount of noise injected into the DP-SGD and PATE-GAN algorithms. The \texttt{max\_per\_sample\_grad\_norm} hyperparameter, which is specific to DP-SGD, sets the gradient clipping threshold.

\begin{table}[H]
\centering
\begin{tabular}{lll}
\toprule
                                & Hyperparameters              & Possible Values             \\
\midrule
\multirow{7}{*}{All models}     & \texttt{epochs}                       & 300                         \\
                                & \texttt{batch\_size}                  & 500                         \\
                                & \texttt{generator\_lr}                & 2e-4, 2e-5                  \\
                                & \texttt{discriminator\_lr}            & 2e-4, 2e-5                  \\
                                & \texttt{generator\_decay}             & 1e-6                        \\
                                & \texttt{discriminator\_decay}         & 1e-6                        \\
                                & \texttt{noise\_multiplier}            & 0.001, 0.1, 1, 5            \\
                                & \texttt{discriminator\_decay}         & 1e-6                        \\
                                & \texttt{batch\_size}                  & 500                         \\
                                & \texttt{noise\_multiplier}            & 0.001, 0.1, 1, 5            \\
\addlinespace
\multirow{2}{*}{DPCTGAN only}   & \texttt{loss}                         & \texttt{cross\_entropy}, \texttt{wasserstein} \\
                                & \texttt{max\_per\_sample\_grad\_norm} & 0.1, 1, 5                   \\
\addlinespace
\\
\multirow{2}{*}{PATECTGAN only} & \texttt{loss}                         & \texttt{cross\_entropy}              \\
                                & \texttt{regularization}               & \texttt{none},  \texttt{dragan}
\\
\bottomrule
\end{tabular}
\end{table}

\subsubsection{Data projection space}

The hyperparameter \texttt{min\_count}, which represents the minimal occurrence data projection (Algorithm~\ref{alg:projection}), is also part of the configuration space with two possible values $\{2, 3\}$ set together with our stakeholders.

\subsection{Cleaning and Constraint Filtering}
\label{sec:components-constraints}

Real-world data often suffer from erroneous values for a variety of reasons, and the Israeli National Registry of Live Births is no exception. Errors may be introduced when values are manually entered into the registry, and artifacts may occur during past database migrations or software updates. Additionally, the raw data may contain a few extreme yet realistic records, where such outlier values can hinder model fitting and exaggerate errors in the evaluation.

These issues are also pertinent to synthetic data generation. Even with clean and ``nice'' original data, the trained generative model might still produce biologically implausible samples because it learns ``soft'' statistical relationships.

Therefore, we established a list of record-level constraints that must be satisfied by \emph{both the raw and synthetic data}. Records not adhering to these constraints are filtered out (1) from the raw data before preprocessing and (2) from the synthetic samples. The initial list was created without access to the raw data, based on published reference charts and established data processing practices in the biostatistics community regarding birth data. \emph{We acknowledge that one author and one primary stakeholder expanded the constraint list after examining the raw data exported from the Ministry of Health database}.

\subsubsection{Raw data constraints}

After exporting the singleton live birth data from the Registry database, we removed rare, extreme, erroneous, or implausible records that might hinder the training of the generative models. In total, less than $1.5\%$ of the records were removed.

\begin{enumerate}
    \item Records with missing values in one field or more
    \item Records with \texttt{birth\_weight} smaller than 500 (strict) OR greater than 5500 (strict)
    \item Records with \texttt{gestation\_week} smaller than 22 (strict) OR greater than 44 (strict)
    \item Records with a \texttt{mother\_age} smaller than 23 (strict) AND a \texttt{parity} greater than 6 (strict)
    \item Records with a \texttt{mother\_age} smaller than 20 (strict) AND a \texttt{parity} greater than 3 (strict)
    \item Records with a \texttt{gestation\_week} smaller than 26 (strict) AND a \texttt{birth\_weight} greater than 1499 (strict)
    \item Records with a \texttt{gestation\_week} smaller than 29 (strict) AND a \texttt{birth\_weight} greater than 2999 (strict)
    \item Records with a \texttt{gestation\_week} smaller than 34 (strict) AND a \texttt{birth\_weight} greater than 3999 (strict)
    \item Records with a \texttt{birth\_weight} smaller than 600 (strict) AND a \texttt{gestation\_week} greater than 29 (strict)
    \item Records with a \texttt{birth\_weight} smaller than 700 (strict) AND a \texttt{gestation\_week} greater than 32 (strict)
\end{enumerate}

\subsubsection{Synthetic data constraints}

The following records were removed from the synthetic data because they were rare, extreme, and implausible, which might exaggerate errors in the evaluation (corresponding to raw data constraints 3, 4, 7, 8, and 9).

\begin{enumerate}
    \item Records with a \texttt{mother\_age} smaller than 23 (strict) AND a \texttt{parity} greater than 6 (strict)
    \item Records with a \texttt{mother\_age} smaller than 20 (strict) AND a \texttt{parity} greater than 3 (strict)
    \item Records with a \texttt{gestation\_week} smaller than 29 (strict) AND a \texttt{birth\_weight} greater than 2999 (strict)
    \item Records with a \texttt{gestation\_week} smaller than 34 (strict) AND a \texttt{birth\_weight} greater than 3999 (strict)
\end{enumerate}

\subsection{Private Selection}
\label{sec:components-private-selection}

The universal scheme (Algorithm~\ref{alg:scheme}), animated by the private selection algorithm (Algorithm~\ref{alg:known-threshold}), essentially implements a differentially private version of random search to find a configuration that produces a dataset that meets all acceptance criteria. This approach is valuable since predicting a suitable configuration without access to the private data is challenging. The private selection algorithm facilitates this process while controlling the total privacy loss budget.

In the private selection algorithm, the mechanism $\M$ corresponds to a single iteration of our scheme, where $x$ represents the release-candidate dataset and $q \in \{0, 1\}$ indicates whether this dataset passed all acceptance criteria, while setting the private selection threshold to $\tau = 1$.
Before we present the formal theorem that provides the privacy parameters of the private selection algorithm, we state it informally first. If $\M$ is $\eps$-DP, then Algorithm~\ref{alg:known-threshold} is $2\eps$-DP for $\gamma \in [0, 1]$ and sufficiently large integer $T$.

The theorem also holds for $\gamma = 0$ and $T = \infty$---meaning the algorithm can run indefinitely until achieving good-enough output while consuming only $2\eps$ privacy loss budget in total. We used these values in our universal scheme.

\begin{theorem}[Private selection with a known threshold \cite{Liu2018PrivateSF}]
\label{thm:known-thresholds}
Fix $\eps_1, \delta_1 > 0, \eps_0 \in [0,1], \gamma \in [0, 1]$. Let $T$ be any integer such that $T \geq \max \Big\{ \frac{1}{\gamma} \ln \frac{2}{\eps_0}, 1 + \frac{1}{e\gamma} \Big\}$, and let $p_1 = \Pr_{q \sim Q(D)}[q \geq \tau]$, then Algorithm~\ref{alg:known-threshold} with these parameters satisfies the following:
\begin{enumerate}
    \item Let $A_{\text{out}}(D)$ be the output of Algorithm~\ref{alg:known-threshold} on input $D$. Then there exists a constant $C$ 
 such that for any  $q \geq \tau$ we have
    \[
    \Pr[A_{\text{out}}(D) = (x, q)] = C \cdot \Pr_{(\tilde{x}, \tilde{q}) \sim Q(D)} [(\tilde{x}, \tilde{q}) = (x, q)].
    \]    
    \item If $\M$ is $\eps_1$-DP, then the output is $(2\eps_1 + \eps_0)$-DP.
    
    \item Let $\tilde{T}$ be the number of iterations of the algorithm, then
    \[
    \E \tilde{T} \leq \frac{1}{p_1(1-\gamma)+\gamma}
    \leq \min \Big\{\frac{1}{p_1}, \frac{1}{\gamma} \Big\}.
    \]
    
    \item Furthermore, $\Pr[A_{\text{out}}(D)= \bot] \leq \frac{(1 - p_1)(1 + \eps_0/2)}{p_1}\gamma$.
\end{enumerate}
\end{theorem}

\begin{algorithm}[tb]
\caption{Private selection with a known threshold algorithm \cite{Liu2018PrivateSF}}\label{alg:known-threshold}
\begin{algorithmic}[1]
\Require a dataset $D$, a differently private algorithm $\M$ that takes a a dataset and returns an output $x$ and a quality score $q \in \R$, a threshold $\tau$, a budget $\gamma \leq 1$ and $\eps_0 \leq 1$, number of steps $T \geq \max \Big\{ \frac{1}{\gamma} \ln \frac{2}{\eps_0}, 1 + \frac{1}{e\gamma} \Big\}$, and sampling access to $Q(D)$.
\For{$j = 1, \ldots, T$}
\State draw $(x, q) \sim \M(D)$
\State if $q \geq \tau$ then output $(x, q)$ and halt;
\State flip a $\gamma$-biased coin: with probability $\gamma$, output $\bot$ and halt;
\EndFor
\State Output $\bot$ and halt.
\end{algorithmic}
\end{algorithm}

\subsection{Software}
\label{sec:components-software}

Releasing official national-level data, even with formal privacy protection, entails significant responsibility. This responsibility is further underscored by known challenges and vulnerabilities associated with implementing differential privacy systems \cite{Mironov2012OnSO,Haney2022PrecisionbasedAA,Stadler2020SyntheticD,Tramr2022DebuggingDP,Jin2021AreWT,Casacuberta2022WidespreadUO,Kifer2020GuidelinesFI,Ganev2023Inadequacy}. In this section, we delineate the overall design of our system, addressing the risks arising from discrepancies between (1) theoretical presentations of our algorithms and (2) their actual code implementation.

\subsubsection{Overall design}

We have developed a Python package called \texttt{synthflow} for this release. Its core functionality is to execute a \emph{flow}, which corresponds to a single iteration of our scheme. Given a configuration, the code transforms the data, trains a generative model, samples records, filters rows based on constraints, applies projection, and performs evaluation according to acceptance criteria. Additionally, it can span the configuration space based on its dimensions and orchestrate the execution of a private selection algorithm. The package also collect the full transcript of all randomness and noise used in the its execution. The logic specific to the implementation details of Israel's National Registry of Live Birth is separate from the flow's general logic.

\subsubsection{Risk management}

The scheme used in this paper fulfil differential privacy with respect to its the theoretical presentation. Implementation of differential privacy algorithms are pruned to bugs that could nullify the theoretical guarantees  \cite{Mironov2012OnSO,Kifer2020GuidelinesFI,Jin2021AreWT,Stadler2020SyntheticD,Haney2022PrecisionbasedAA,Tramr2022DebuggingDP,Casacuberta2022WidespreadUO}.

We did not have the resources to deploy end-to-end formal verification or conduct a complete differential privacy auditing. Given the available resources to us, we aimed to manage and reduce the risk of misalignment between the implementation and the theoretical framework presented in this paper. The following points summarize the measures we took.  

\paragraph{Keep it simple.} As a general principle, we preferred simple code design that is easier to test and maintain.

\paragraph{Applying best practices of software engineering.} Testing, linting, documentation and reviewing were integral to the development workflow. While insufficient to identify potential bugs with differential privacy, these practices are the first line measure for mitigating risks. In particular, we carefully examine each manipulation and query on the original data (e.g., definition of columns' boundaries) and calculations of parameter values.

\paragraph{Using differentially private algorithm from established open source packages; not implementing new algorithms.} Because it requires great effort to implement differential privacy mechanisms correctly with a production quality level, we avoided writing our own implementation. We opt to use existing open source packages that are either well and actively maintained, SmartNoise\footnote{Synthesizers: \url{https://github.com/opendp/smartnoise-sdk}}\footnote{We patched the MWEM implementation in SmartNoise to allow the exponential mechanism to choose the same query; because otherwise, the code got into an infinite loop on the NVSS data)} which is part of OpenDP ecosystem \cite{Gaboardi2020APF} and Diffprivlib\footnote{Laplace mechanism and private linear regression; \url{https://github.com/IBM/differential-privacy-library}} \cite{Holohan2019DiffprivlibTI}, or that are used by other projects and undergone a code review by us (PrivBayes from SDGym\footnote{ \url{https://github.com/sdv-dev/SDGym/tree/c9e274c1c1be7e8fec6fcd1d6f88e95b38a44d14/privbayes}}). \emph{Nonetheless, these packages are susceptible to floating point vulnerabilities\footnote{The Diffprivlib package has the implementation of the Snapping mechanism \cite{Mironov2012OnSO}, but due to dependency issue, it could not work as-is in the OS used in the enclave environment (Windows).}\footnote{At the time of producing the released data, SmartNoise used NumPy to generate the Laplace noise.} and underestimation of the sensitivity vulnerabilities, and consequently, also our release.} This issue is one of the reasons why we decided against releasing the generative model with the full-precision float probabilities, opting instead to make only the final dataset publicly available. We plan to resolve this issue in future releases of the Registry.

\paragraph{Perform only essential patching of code for differential privacy mechanisms.} We patched the PrivBayes code for three purposes.

First, we addressed the bug reported by \textcite{Stadler2020SyntheticD}: we now take the upper and lower bounds of each column as predefined inputs from the configuration, rather than calculating them directly from the original dataset in a non-private manner.

Second, we switched to a CSPRNG (cryptographically secure pseudo-random number generator) as discussed in Section~\ref{sec:components-software-randomness}.

Third, we added functionality to collect the transcript of all random and noise used in fitting the Bayesian network. However, this transcript is neither exported from the enclave nor released, but may be used for future research.

\paragraph{Considering synthetic data and data transformations as as privacy protection fallbacks, but not as privacy guarantees.} 
Our privacy protection in this release is given by differential privacy. 
Synthetic data and transformations are not considered privacy guarantees, but by adopting a defense-in-depth approach, they serve as a secondary layer of protection by limiting the precision and amount of information released.
We stress that the dataset projection for face privacy is not considered a fallback.

\subsubsection{Randomness}
\label{sec:components-software-randomness}

High-quality randomness is essential for the differential privacy guarantee to hold---a requirement it shares with cryptographic applications.
For an in-depth discussion on the role of randomness in differential privacy and a detailed comparison with cryptography, refer to \textcite{Garfinkel2020RandomnessCW}.

To meet this requirement, we ensured that the implementations of differentially private mechanisms use high-quality Cryptographically Secure Pseudo-Random Number Generators (CSPRNG): (1) the PrivBayes algorithm; (2) the Laplace and Functional mechanisms for acceptance criteria; and (3) configuration sampling in the private selection algorithm. It is best practice to use the operating system's CSPRNG. As the release produced in a Windows-based enclave environment (see Section~\ref{sec:components-environments}), the random generator utilized is the Cryptographic Service Provider \cite{ms2021crypto}.

To utilize Windows' CSPRG, we modified the PrivBayes implementation to use \texttt{random\_device} from the \texttt{boost} library as its random generator, instead of the standard C library's insecure generator from \texttt{random.h}. For the acceptance criteria and the private selection algorithm, the implementations make use of the \texttt{secrets} package from Python's standard library.

We leave to future research to explore the properties required from a random generator for the secure deployment of differential privacy.

\subsection{Execution Environments}
\label{sec:components-environments}

The execution of our scheme was conducted in an enclave environment managed by TIMNA, Israel’s National Health Research Platform. Only the final release-candidate dataset with the differentially private results of its acceptance criteria were exported from the environment.

The experiments on the public data (see Section~\ref{sec:intro-public-data}) were run on Boston University Shared Computing Cluster (SCC)\footnote{\url{https://www.bu.edu/tech/support/research/computing-resources/scc}}.

\subsection{Public-facing Documentation}
\label{sec:components-documentation}

In addition to this paper, it's essential to communicate the release details and usage guidelines to various audiences, particularly data subjects and data users. Communication should cater to the diverse interests, needs, and knowledge levels of these groups.

We've created a comprehensive README document (\texttt{pdf} format), developed in collaboration with the stakeholders. This document accompanies the data release (\texttt{csv} format) and is available via the Israeli Government Open Data Portal \cite{moh2024release}.

The README structure draws from Open Government Data best practices \cite{tauberer2014open} and dataset transparency initiatives like Datasheets for Datasets \cite{GebruMVVWDC21}, Dataset Nutrition Label \cite{HollandHNJC18}, and Data Cards \cite{PushkarnaZK22}. It consists of two parts: The first part addresses both data subjects and data users, offering metadata, intended usage, and a request for feedback. The second part focuses on technical details for data users, including data quality, privacy considerations, and production process.

As discussed in Section~\ref{sec:intro-transparency}, while the released dataset maintains the same affordances as the original dataset, it has been specifically designed and validated for predefined statistical queries outlined in the acceptance criteria. To mitigate the risk of inappropriate analysis, we provide clear guidance through the \emph{intended usage} section of the README document.

We plan to produce additional materials for the general public and data subjects in Israel in a more engaging way, such as short videos about the potential utility of releasing this data to the public and the measures taken to protect the privacy of the data subjects.

\subsection{Post-Production Improvements}
\label{sec:components-improvements}

In this section, we discuss potential enhancements to our implementation identified after the released dataset was produced. The composition property of differential privacy means that re-running our scheme would exceed the total privacy loss budget that the stakeholders and we were willing to accept. These enhancements would have minimal or no effect on the privacy guarantee, with one enhancement potentially saving $\eps = .43$ of the privacy loss budget. Thanks to the acceptance criteria, the data quality, as operationalized by the stakeholders, remained unaffected.

\subsubsection{Saving privacy budget for linear regression acceptance criteria}
\label{sec:components-improvements-lr}

In the implementation of the linear regression acceptance criteria, we calculated the linear regression coefficients for the transformed dataset in each iteration of the scheme, even though it is not dependent on the configuration (Section~\ref{sec:components-acceptance-criteria-regression}). We could compute the differentially private linear regression just once before initiating the loop in our scheme. This change would save $\eps = .43$ in the privacy loss budget by circumventing the doubling of the privacy loss budget induced by the private selection algorithm (Theorem~\ref{thm:known-thresholds}). It's important to note that this improvement does not apply to other acceptance criteria since their computations directly involve the transformed and synthetic datasets.

\subsubsection{Wrong global sensitivity calculation for MAE acceptance criterion}
\label{sec:components-improvements-mae}

In the implementation of the Mean Absolute Error (MAE) acceptance criteria, we used a global sensitivity that was half the correct value (Section~\ref{sec:components-acceptance-criteria-regression}). The MAE was perturbed using the Laplace mechanism, where the added noise depends on the ratio between global sensitivity $\Delta$ and the privacy loss budget $\eps$: $\Delta/\eps$. To rectify the incorrect sensitivity value, we doubled the allocated privacy loss budget from $\eps = .02$ to $\eps = .04$. This adjustment would have been made even if the sensitivity had been correctly calculated, since the privacy loss budget was determined based on the standard deviation of the noise, ensuring the accuracy of the evaluation. 

The error has been corrected in the code accompanying with this paper.

\subsubsection{Simpler dataset projection algorithm}
\label{sec:components-improvements-projection}

The algorithm for dataset projection described in Section~\ref{sec:components-projection} is an improved version of the actual algorithm used to produce the release dataset. We realized that it could be improved only after consuming the privacy loss budget for running our scheme on the Live Birth Registry data.

In this section we describe the implemented algorithm (Algorithm~\ref{alg:projection-legacy}) and its properties in Proposition~\ref{prop:projection-legacy}, which is very similar to Proposition~\ref{prop:projection}.

The use of Algorithm~\ref{alg:projection-legacy} instead of Algorithm~\ref{alg:projection} has no impact on privacy as well as the data quality with respect to the acceptance criteria.

\begin{algorithm}
\caption{Dataset Projection for Minimal Occurrence (Implementation Version)}\label{alg:projection-legacy}
\begin{algorithmic}[1]
    \Require a multiset (dataset) $\DD$, a minimum count of records $m$.
    \State $R_0 \leftarrow \varnothing$
    \For{$k = 1, \ldots, m-1$}
        \State $R_k' \leftarrow R_{k-1} \cup \REC(\DD, \#=k)$
        \State $R_k \leftarrow$ Sample without replacement a proportion of $\floor{\frac{k}{k+1}}$ records from $R_k'$
    \EndFor
    \State $\DD'_{< m} \leftarrow \{(x, m): x \in R_{m-1} \}$
    \State $\DD'_{=m} \leftarrow \{(x, \DD(x)): x \in \XX, \DD(x) = m \}$
    \State $\DD'_{> m} \leftarrow \{(x, \DD(x)): x \in \XX, \DD(x) > m \}$
    \State $\DD' \leftarrow \DD'_{< m} \cup \DD'_{=m} \cup \DD_{> m}'$
    \State \Return $\DD'$
\end{algorithmic}
\end{algorithm}

\begin{proposition}
    \label{prop:projection-legacy}
    Let $\DD$ be a multiset (dataset) and $\DD'$ be the multiset corresponding to the output of Algorithm~\ref{alg:projection-legacy} when running it on $\DD$ with a minimum count of records $m$. Assume $\forall k \in [m-1]: \frac{m!}{k!}| n^\DD_k$. Then the following holds:
    \begin{enumerate}
        \item $\REC(\DD') \subseteq \REC(\DD)$,
        \item For all $k > m$, $\REC(\DD',\# = k) = \REC(\DD,\# = k)$,
        \item $\REC(\DD,\# = m) \subseteq \REC(\DD',\# = m)$,
        \item $\REC(\DD',\# < m) = \varnothing$,
        \item For all $k < m$, $\E \Big [\sum\limits_{x \in \REC(\DD,\#=k)} \DD'(x) \Big] = \sum\limits_{x \in \REC(\DD,\#=k)} \DD(x)$
        \item $|\DD'| = |\DD|$,
    \end{enumerate}
\end{proposition}

\begin{proof}
    Statements (1)-(4) holds according to the same arguments as in Proposition~\ref{prop:projection}.

    The proof of statement (6) here is parallel to its proof in Proposition~\ref{prop:projection}. It differs in showing that $|\DD_{<m}'| = \sum\limits_{x \in \REC(\DD,\#<m)} \DD(x)$. Note that $|\DD_{<m}'| = m \cdot |R_{m-1}|$ according its definition. The following lemma helps us to calculate  $|R_{m-1}|$.

\begin{lemma}
\label{lmm:projection-legacy}
    If $\forall k \in [m-1]: \frac{m!}{k!}| n^\DD_k$, then
        \[
        |R_k| = \frac{1}{k+1} \sum_{t = 1}^k t \cdot n_t^\DD.
        \]
\end{lemma}

\begin{proof}
    By induction on $k$. For $k = 1$, $|R_1| = \frac{1}{2}|\REC(\DD, \#=1)| = \frac{1}{2} n^{\DD}_1 $ by the definition of $R_1$ and $n^{\DD}_1$.

    We assume that the claim holds for $k-1$ and prove it for $k$. By definition, $R_k$ is a $\frac{k}{k+1}$ proportion of $R_k' = R_{k-1} \cup \REC(\DD, \#=k)$, a union of two disjoint sets, so
   
    \[
        |R_k| = \frac{k}{k+1}|R_k'| = \frac{k}{k+1} \Big(|R_{k-1}| + |\REC(\DD, \#=k+1)| \Big) = \frac{k}{k+1} |R_{k-1}| + \frac{k}{k+1} n_k^\DD.
    \].

    Plugging in the induction hypothesis, we get the required
    
    \[
        |R_k| = \frac{k}{k+1} \cdot \frac{1}{k} \sum_{t = 1}^{k-1} t \cdot n_t^\DD + \frac{k}{k+1} n_k^\DD
        = \frac{1}{k+1} \Big( \sum_{t = 1}^{k-1} t \cdot n_t^\DD + k \cdot n_k^\DD \Big)
        = \frac{1}{k+1} \sum_{t = 1}^{k} t \cdot n_t^\DD.
    \].

    Note that we can ignore the floor function $\floor{\,\cdot\,}$ in the sampling from $R'_{k}$ thanks to the lemma condition.
   
\end{proof}

    Consequently,
    \[
        |\DD_{<m}'|
        = m \cdot |R_{m-1}| = m \cdot \frac{1}{m} \sum_{k = 1}^{m-1} k \cdot n_k^\DD
        = \sum_{k = 1}^{m-1} k \cdot n_k^\DD
        = \sum\limits_{x \in \REC(\DD,\#<m)} \DD(x).
    \]

    For statement (5) note that for $k < m$, the probability of including an $x \in \REC(\DD, \#=k)$ in $\DD'$ is $\frac{m-1}{m} \cdot \frac{m-2}{m-1} \cdot \ldots \cdot \frac{k+1}{k+2} \cdot \frac{k}{k+1} = \frac{k}{m}$; and if $x$ is included, then $\DD'(x) = m$, otherwise $\DD'(x) = 0$.
    
    Therefore,

    \begin{align*}
        \E \Big [\sum\limits_{x \in \REC(\DD,\#=k)} \DD'(x) \Big]
        &= \sum\limits_{x \in \REC(\DD,\#=k)} \E \Big [\DD'(x) \Big]\\
        &= \sum\limits_{x \in \REC(\DD,\#=k)} \Big( \frac{k}{m} \cdot m + \frac{m-k}{m} \cdot 0 \Big) \\
        &= \sum\limits_{x \in \REC(\DD,\#=k)} k \\
        &= n_k^{\DD} \cdot k
        = \sum\limits_{x \in \REC(\DD,\#=k)} \DD(x).
    \end{align*}

\end{proof}

%% file: sections/3-related-work.tex
\section{Related Work}
\label{sec:related-work}

\paragraph{Differentially private synthetic data generation.} \textcite{Tao2021BenchmarkingDP} proposed a categorization of differential privacy synthetic data algorithms: (1) Marginal-based methods, which privately select and measure a collection of marginals before generating records from the induced distribution; (2) Workload-based methods that iteratively and adaptively refine the model based on the error of a set of queries; and (3) GAN-based methods that employ a generative adversarial network (GAN) with differential privacy. Among these, marginal-based algorithms like AIM \cite{McKenna2022AIMAA}, MST \cite{McKenna2021WinningTN}, and PrivBayes \cite{Zhang2014PrivBayesPD} perform best in similar settings like ours according to comprehensive evaluations \cite{Tao2021BenchmarkingDP,ganev2023understanding} and recent NIST competition results \cite{nist18synth,nist20temporal}. \textcite{Liu2021IterativeMF} developed a framework that unifies the presentation of workload-based algorithms, including MWEM \cite{Hardt2010ASA}, DualQuery \cite{Gaboardi2014Dual}, PEP, and RAP \cite{Brown2021RAP}. Various GAN variants exist, either employing DP-SGD \cite{Abadi2016DPSGD} or the PATE framework \cite{nicolas2018scalable} for training non-private GAN models for tabular data \cite{Xie2018DPGAN,Jordon2019PATEGAN,Rosenblatt2020DifferentiallyPS}.

\paragraph{Differential Privacy real-world releases.} \textcite{desfontain2021ListRealworld} offers an extensive list of real-world data releases and systems employing differential privacy. To the best of our knowledge, this list is the most current and well-maintained. Here, we highlight only a few deployments that have significantly influenced or inspired this release. Most differential privacy deployments have been carried out by the US Government or big tech companies. For instance, the US Census has released several data products incorporating differential privacy, notably the 2020 Census Redistricting Data \cite{Abowd2020Topdown}. Other releases include the Post-Secondary Employment Outcomes of college graduates \cite{David2019PSEO} as well as OnTheMap \cite{Machanavajjhala2008Map}, which marked the first real-world application of differential privacy. The Wikimedia Foundation, in collaboration with Tumult Labs, recently published deferentially private statistics detailing daily visits to Wikipedia pages by country \cite{Adeleye2023Wikimedia}.
During the COVID-19 pandemic, Google consistently released differentially private aggregate statistics on changing trends in mobility patterns over time, segmented by geography and across various categories of places \cite{Aktay2020COVID}.
Only a single release in the list of \textcite{desfontain2021ListRealworld} uses synthetic data: the Global Victim-Perpetrator dataset prepared in collaboration between the International Organization for Migration and Microsoft \cite{Migration2022Global}. The dataset provides first-hand information on the relationships between victims and perpetrator in human trafficking.

\paragraph{Communication of Differential Privacy guarantees.}

In recent years, the question of how to effectively communicate differential privacy guarantees to various stakeholders has received significant attention, recognized as an essential factor for successful DP releases. Studies based on surveys and interviews have been conducted to assess the intelligibility and perception of different communication strategies for DP and, consequently, the willingness of users to share information. \textcite{Cummings2021INA} proposed a framework to evaluate how descriptions of DP influence this willingness. \textcite{Franzen2022AmIP} designed quantitative risk notifications based on communication formats from the medical field. \textcite{Xiong2022Illustrations} used illustrations to explain three modes of DP: central, local, and shuffle, and \textcite{Karegar2022Metaphors} examined the effectiveness of metaphors. \textcite{Voigt2024FromTheory} explored how well users understood DP explanations, using an explanation of $k$-anonymity as a baseline. \textcite{Nanayakkara2023WhatAT} focused on explaining the value of the privacy parameter $\eps$.

%% file: sections/4-discussion.tex
\section{Discussion}
\label{sec:discussion}

In this work, we documented the development of a differentially private release of the Israeli National Registry of Live Births in 2014. We identified four requirements elicited from our stakeholders: (1) tabular format; (2) data quality; (3) faithfulness; and (3) privacy. We designed a universal scheme that produces a release that satisfies these requirements and enjoys an end-to-end differential privacy guarantee. We have identified the following key directions for future research.

First, we were fortunate to find similar public data to mark which configurations have a higher chance of passing the acceptance criteria. However, not for every dataset a similar similar public data might be available. Note that the public data is not used here to seed the learning algorithm for the generative model. What could be done if we don't have access to public data? We hypothesize that, for the purpose of narrowing down the space of configurations, only partial knowledge about the target distribution (e.g., in the format of statistical facts) might be sufficient.

Second, data quality statements, as expressed via acceptance criteria, are significantly valued by the data users. We anticipate that data users would like to run additional queries beyond what is covered by our criteria. Is it possible to bound the error of other queries given the result of the acceptance criteria we released? How could one choose a list of differentially private acceptance criteria that ``span'' a more extensive set of bounds on other desired queries?

Third, most of our acceptance criteria consider maximum error. We chose this form of criteria because it provides a bound on any other type of error. Perhaps users would be interested in different kinds of bounds, such as average error, confidence intervals, or tail error (like faithfulness). We hope to collect feedback from the public and public health researchers about the structure and content of the current release, and in particular, the desired acceptance criteria from the data users.

Fourth, the meaning of the differential privacy guarantee must be understood within the context of the release. For example, which $\eps$ should a data curator use? How should a data subject understand the protection of $\eps = 9.98$ \cite{Garfinkel2018IssuesED,Wood2018DifferentialPA,Cummings2021INA,Franzen2022AmIP,Nanayakkara2023WhatAT}? In this work, we adopted a heuristic approach for choosing the value of the privacy parameter based on previous releases \cite{desfontain2021ListRealworld}, but a more principled and interpretable method is needed.

Fifth, the released dataset is designed to accurately answer only a subset of statistical queries, such as $k$-way marginals. However, the affordance of the microdata format ``invites'' users to run other types of queries without data quality empirical guarantees. Assuming a tabular format is a non-negotiable requirement, what measures should a curator take to prevent out-of-band analysis?

Sixth, when the government releases national data, issues of trust should be addressed meticulously. For instance, in this release, we incorporated the principles of faithfulness and face privacy to serve this purpose. How can the data curator assure a distrustful stakeholder that the release was produced as claimed? What additional properties should the release possess to promote trust?

%% file: sections/5-acknowledgement.tex
\section*{Acknowledgement}
\addcontentsline{toc}{section}{Acknowledgement}

We would like to express our sincere gratitude to Meytal Avgil Tsadok and Roy Cohen from TIMNA (Israel’s National Health Research Platform, Ministry of Health) for initiating this project, and for their vision, trust, courage, flexibility, dedication, and unwavering support, all of which were crucial to its successful completion.
We also extend our thanks to Orly Manor, Ziona Haklai, and Deena Zimmerman for their extremely helpful feedback, and to Barak Shukrun for his support. We are especially grateful to Adam Smith for his invaluable advice, insightful feedback, and keen eye for flaws in our design. We are also grateful for the guidance provided by Andrew Sellars and his team at the BU/MIT Technology Law Clinic. Special thanks to Iden Kalemaj, Gavin Brown, Niva Elkin-Koren, and Salil Vadhan for helpful discussions.

This material is based upon work supported by DARPA under Agreement No. HR00112020021. Any opinions, findings and conclusions or recommendations expressed in this material are those of the author(s) and do not necessarily reflect the views of the United States Government or DARPA.